\documentclass[12pt]{article}
\usepackage{graphicx,amsfonts,amsmath,amssymb,amsthm,mathrsfs,mathabx,url,hyperref}
\usepackage[usenames,dvipsnames]{xcolor}
\usepackage[toc,page]{appendix}

\title{Implementing Bogoliubov Transformations Beyond the Shale--Stinespring Condition}
\author{
Sascha Lill\footnote{Universit\`a degli Studi di Milano, Dipartimento di Matematica, Via Cesare Saldini 50, 20133 Milano, Italy}\ \footnote{E-Mail: sascha.lill@unimi.it}~
} 
\date{\today}

\usepackage[utf8]{inputenc}
\usepackage{tikz}
\usetikzlibrary{decorations.pathreplacing}

\newcommand*\circled[1]{\tikz[baseline=(char.base)]{
            \node[shape=circle,draw,inner sep=0, outer sep = 0] (char) {#1};}}	

\usepackage[section]{placeins}
\usepackage{capt-of, caption}
\usepackage{paralist}
\usepackage{booktabs}
\usepackage{hhline}
\usepackage{color}                  

\newcommand{\be}{\boldsymbol{e}}
\newcommand{\bof}{\boldsymbol{f}}
\newcommand{\bg}{\boldsymbol{g}}
\newcommand{\bh}{\boldsymbol{h}}

\newcommand{\bp}{\boldsymbol{p}}
\newcommand{\br}{\boldsymbol{r}}
\newcommand{\bs}{\boldsymbol{s}}

\newcommand{\bx}{\boldsymbol{x}}

\newcommand{\boF}{\boldsymbol{F}}
\newcommand{\bG}{\boldsymbol{G}}

\newcommand{\boeta}{\boldsymbol{\eta}}

\newcommand{\bphi}{\boldsymbol{\phi}}

\newcommand{\cA}{\mathcal{A}}
\newcommand{\cD}{\mathcal{D}}
\newcommand{\cE}{\mathcal{E}}
\newcommand{\cI}{\mathcal{I}}

\newcommand{\cM}{\mathcal{M}}

\newcommand{\cQ}{\mathcal{Q}}
\newcommand{\cS}{\mathcal{S}}

\newcommand{\cV}{\mathcal{V}}
\newcommand{\fc}{\mathfrak{c}}
\newcommand{\fh}{\mathfrak{h}}
\newcommand{\fr}{\mathfrak{r}}

\newcommand{\sF}{\mathscr{F}}
\newcommand{\sH}{\mathscr{H}}

\newcommand{\CCC}{\mathbb{C}}

\newcommand{\NNN}{\mathbb{N}}

\newcommand{\RRR}{\mathbb{R}}
\newcommand{\UUU}{\mathbb{U}}
\newcommand{\ZZZ}{\mathbb{Z}}

\newcommand{\crit}{{\rm crit}}
\newcommand{\Cseq}{{\rm Cseq}}

\newcommand{\di}{{\rm d}}
\newcommand{\dom}{{\rm dom}}
\newcommand{\eRen}{\mathrm{eRen}}
\newcommand{\ex}{{\rm ex}}

\newcommand{\Pol}{\mathrm{Pol}}

\newcommand{\reg}{{\rm reg}}

\newcommand{\sgn}{\mathrm{sgn}}
\newcommand{\Ren}{\mathrm{Ren}}
\newcommand{\vecspan}{\mathrm{span}}

\newcommand{\tr}{\mathrm{tr}}
\newcommand{\uni}{{\rm uni}}

\newcommand{\FB}{\overline{\mathscr{F}}}
\newcommand{\cAB}{\overline{\mathcal{A}}}
\newcommand{\Sdot}{\dot{\mathscr{S}}^\infty}
\newcommand{\Qdot}{\dot{\mathcal{Q}}}

\newtheorem{theorem}{Theorem}[section]
\newtheorem{lemma}[theorem]{Lemma}

\newtheorem{proposition}[theorem]{Proposition}

\theoremstyle{definition}
\newtheorem{remark}[theorem]{Remark}
\newtheorem{definition}[theorem]{Definition}

\begin{document}
\maketitle
\begin{abstract}
We define infinite tensor product spaces that extend Fock space, and allow for implementing Bogoliubov transformations which violate the Shale or Shale--Stinespring condition. So an implementation on the usual Fock space would not be possible. Both the bosonic and fermionic case are covered. Conditions for implementability in an extended sense are stated and proved. From these, we derive conditions for a quadratic Hamiltonian to be diagonalizable by a Bogoliubov transformation that is implementable in the extended sense.
We apply our results to Bogoliubov transformations from quadratic bosonic interactions and BCS models, where the Shale or Shale--Stinespring condition is violated, but an extended implementation nevertheless works.\\

\medskip

\end{abstract}

\newpage
\tableofcontents
\newpage

\section{Introduction}	
\label{sec:intro}

In quantum many-body systems and non-perturbative Quantum Field Theory (QFT), one often encounters situations in which a formal expression $ H $ for a Hamiltonian is given in the physics literature, but can a priori not be interpreted as a self-adjoint operator that generates dynamics on a Hilbert space. Over the past decades, a plethora of mathematical tools has been conceived to overcome this issue \cite{Dedushenko2022, Summers2012}.\\
An established but little investigated tool in this area are Fock space extensions, such as the \textbf{infinite tensor product} (ITP) framework, introduced by von Neumann \cite{vonNeumann1939}. There have been attempts to apply this framework to Quantum Electrodynamics (QED) scattering theory \cite{Chung1965, Kibble1986, KulishFaddeev1970}. Rigorous results exist on the implementation of Weyl transformations (i.e., linear exponents) that are not implementable on the usual Fock space \cite{Blanchard1969, Froehlich1973, KoenenbergMatte2014}.\\

The present paper takes the step from linear to quadratic exponents, i.e., from Weyl to Bogoliubov transformations: We derive conditions which ensure that a Bogoliubov transformation, which is not implementable on the usual Fock space, can nevertheless be implemented on a suitable ITP space.\\
Non-implementable Weyl- and Bogoliubov transformations are both a well-researched topic in the $ C^*$-algebraic formulation of quantum dynamics \cite{DerezinskiGerard2013}. Therefore, they serve as an ideal test area for mathematical tools that may ultimately turn out advantageous for more sophisticated operator transformations.\\

Let us explain a bit more precisely, how the implementation of Bogoliubov transformations via ITPs works. Roughly speaking, a Bogoliubov transformation $ \cV = \left( \begin{smallmatrix} u & v \\ \overline{v} & \overline{u} \end{smallmatrix} \right) $ replaces annihilation operators $ a^\dagger(f), a(f) $ by
\begin{equation*}
	 b^\dagger(f) = a^\dagger(u f) + a(v \overline{f}) \;, \qquad
	 b(f) = a^\dagger(v \overline{f}) + a(u f) \;,
\end{equation*}
with $ f $ being an element of the one-particle Hilbert space $ \fh $, where $ u,v $ are linear operators on $ \fh $ and where $ \overline{f} $ is the complex conjugate of $ f $. This replacement may diagonalize quadratic Hamltonians $ H $ \cite{Araki1968, BachBru2016, NamNapiorkowskiSolovej2016, MatsuzawaSasakiUsami2021}. Related transformations allow for eliminating inconvenient terms of higher order in more sophisticated $ H $ \cite{Schlein2019, Giacomelli2022short, Giacomelli2022, FalconiGiacomelliHainzlPorta2021}.\\
It is desirable to find a unitary operator $ \UUU_\cV $ on Fock space $ \sF $, such that $ \UUU_\cV $ establishes the replacement $ a^\sharp \mapsto b^\sharp $ via
\begin{equation}
	\UUU_\cV a^\dagger(f) \UUU_\cV^* = b^\dagger(f) \;, \qquad
	\UUU_\cV a(f) \UUU_\cV^* = b(f) \;.
\end{equation}
In that case we say that $ \UUU_\cV $ implements the transformation $ \cV $ and we call $ \cV $ ``implementable'' (in the regular sense). It is well-known \cite{Shale1962, Ruijsenaars1977} that $ \cV $ is implementable, if and only if the {\bf Shale condition} (bosonic case) or the {\bf Shale--Stinespring condition} (fermionic case) holds, which asserts that $ \tr(v^*v) < \infty $.
Situations with non-implementable Bogoliubov transformations occur, for instance, in relativistic models \cite{TorreVaradarajan1999, Marecki2003}, and within many-body systems of infinite size \cite{BardeenCooperSchrieffer1957, Haag1962, Derezinski2017}.\\

We prove that under certain conditions, $ \cV $ is nevertheless implementable on certain ITP spaces $ \widehat{\sH} = \prod_{k \in \NNN}^\otimes \sH_k $, which extend all of $ \sF $, while being a sum of typically uncountably many spaces that are naturally isomorphic to $ \sF $. We further elucidate the structure of $ \widehat{\sH} $ in the end of Section \ref{subsec:infinitetensorprod}, as well as in Appendix~\ref{sec:resultsITP}.\\

Also, the way how we diagonalize formal Hamiltonians $ H $ by implementers $ \UUU_\cV $ differs from the usual procedure on Fock space: For some formal $ H $ consisting of a product of $ a^\dagger $- and $ a $-operators, we aim at defining (see Figure \ref{fig:Dressing_Fockspace})
\begin{equation}
	\widetilde{H} = \UUU_\cV^{-1} (H + c) \UUU_\cV
\label{eq:renormalizedhamiltonian}
\end{equation}
on a dense subspace $ \cD_\sF \subseteq \sF $, where $ \widetilde{H} $ is the version of $ H $ with $ a^\sharp $ replaced by $ b^\sharp $ and with normal ordering applied. The ``renormalization constant'' $ c $ stems from normal ordering and can be infinite. While \eqref{eq:renormalizedhamiltonian} may not always be achieved, we can still split $ H = \sum_{n \in \NNN} H^{(n)} $ and $ c = \sum_{n \in \NNN} c^{(n)} $, such that
\begin{equation}
	\widetilde{H} = \sum_{n \in \NNN} \UUU_\cV^{-1} (H^{(n)} + c^{(n)}) \UUU_\cV \;.
\label{eq:renormalizedhamiltonian2}
\end{equation}
That is, we define an operator $ \UUU_\cV: \cD_\sF \to \widehat{\sH} $ such that $ (H^{(n)} + c^{(n)}) $ maps the space $ \UUU_\cV[\cD_\sF] \subset \widehat{\sH} $ into itself. Further, $ \widetilde{H} $ allows for a self-adjoint extension and can thus be seen as the well-defined renormalized version of $ H $.
\begin{figure}
	\centering
	\scalebox{1.0}{\begin{tikzpicture}
 \filldraw[thick,fill = blue!10!white] (4,-0.2) rectangle ++ (2,1.4);
 \filldraw[thick,fill = gray!30!white] (4.9,0.5) ellipse (0.8 and 0.4) node {$ \mathcal{D}_\mathscr{F} $};
 \draw[blue] (4.5,1.2) -- ++(0.5,0.3) node[anchor = south] {Fock space $ \mathscr{F} $};

 \filldraw[fill = blue!05!white] (8.7,-0.3) rectangle ++(4.5,2.3);
 \filldraw[thick,fill = blue!10!white] (9,-0.2) rectangle ++ (2,1.4);
 \filldraw[thick,dashed,fill = gray!30!white] (12.2,1.2) ellipse (0.8 and 0.4) node {$ \mathbb{U}_{\mathcal{V}}[\mathcal{D}_\mathscr{F}] $};
 \filldraw[thick,fill = gray!30!white] (9.9,0.5) ellipse (0.8 and 0.4) node {$ \mathcal{D}_\mathscr{F} $};
 \draw[line width = 2,->] (11.5,1.4) .. controls (11.2,1.5) and (10.8,1.4) .. (10.5,0.8);
 \node at (10.8,1.6) {$ \mathbb{U}_{\mathcal{V}}^{-1} $};
 \draw[line width = 2,<-] (11.5,1.2) .. controls (11.2,1.3) and (10.8,1.1) .. (10.6,0.6);
 \node at (11.3,0.8) {$ \mathbb{U}_{\mathcal{V}} $};
 \draw[line width = 2,->] (12.7,0.9) .. controls (13.8,0.5) and (13.8,1.6) .. (12.8,1.4);
 \node at (13.6,1.7) {\scriptsize $ (H^{(n)} \! + c^{(n)}) $};
 \draw[blue] (13.2,0) -- ++(0.4,0.3) node[anchor = west] {$ \widehat{\mathscr{H}} $};

\end{tikzpicture}}
	\caption{Renormalization of $ H $ using ITPs.}
	\label{fig:Dressing_Fockspace}
\end{figure}

Our {\bf main result} is that in the extended sense, specified in Definition \ref{def:implementation}, $ \cV $ can indeed be implemented on $ \widehat{\sH} $ in the bosonic (Theorem \ref{thm:bosoniccountable}) and the fermionic case (Theorem \ref{thm:fermioniccountable}) if the spectrum of the operator $ v^*v $ is countable. We may then naturally extend the implementer to a unitary operator $ \UUU_\cV: \widehat{\sH} \to \widehat{\sH} $, see Remark \ref{rem:ITPunitary}.\\

It is worth mentioning that within the \textbf{extended state space} (ESS) framework~\cite{lill}, a similar result was recently proven even for arbitrary $ v^*v $ in the bosonic case~\cite{lillproc}. In Appendix~\ref{app:ESS}, we briefly discuss a much simpler ESS construction for $ v^*v $ having discrete spectrum, which allows for an extended implementation in both the bosonic and fermionic case (Propositions~\ref{prop:bosoniccountableESS} and~\ref{prop:fermioniccountableESS}).\\
With respect to ITPs, the fermionic ESS implementation requires the additional restriction to a finite number of particle--hole transformed modes. Nevertheless, we expect a fermionic ESS construction, similar to~\cite{lillproc}, to also be successful for generic $ v^*v $.\\
By contrast, the ITP construction cannot be expected to achieve an implementation for generic $ v^*v $: Still, the ITP space would be well-defined as $ \prod_{x \in X}^\otimes \sH_x $ with $ X $ being a possibly uncountable set related to $ \sigma(v^*v) $. However, $ a^*(f) = \sum_x f(x) a^*_x $ would only be defined for countable sums in $ x $. So $ a^*(f) $ would be ill-defined unless $ f(x) $ is everywhere 0 apart from countably many $ x $.\\

The ultimate goal would be to implement more general operator transformations $ W $ such that
\begin{equation}
	\widetilde{H} = W^{-1} (H + c) W \;,
\end{equation}
is well-defined on $ \cD_\sF \subseteq \sF $ and allows for a self-adjoint extension. Here, $ c $ is a general counterterm and not necessarily just a constant. Transformations $ W $ as above arise from non-perturbative cutoff renormalization \cite{Nelson1964, Glimm1968, Derezinski2003, GlimmJaffeI, GlimmJaffeY2I, Gross1973}, when formally removing the IR- or UV-cutoff from the employed dressing transformations. In contrast to cutoff renormalization, the direct renormalization by Fock space extensions does not involve limiting processes or cutoffs that break Lorentz invariance. A similar cutoff-free non-perturbative renormalization technique, also known as ``interior--boundary conditions'' (IBC), has recently been proposed and investigated \cite{I00b, I02, I03, I04, I06, I05, I11, I10, I07}. However, IBC renormalization is limited to cases where the free and interacting Hamiltonian can be defined in the same CCR/CAR representation. This can be a severe restriction, for instance, in relativistic models\footnote{In relativistic QFT, Haag's theorem forbids a common representation for free and interacting dynamics \cite[Sect.~II.1]{Haag1996}. But such a common representation may also fail to exist in more general situations, like in polaron models \cite{Gross1973} or when taking the thermodynamic limit of many-body systems \cite[Ch.~17]{DerezinskiGerard2013}.}, and renormalization by Fock space extensions is designed to overcome this issue.\\

The rest of this paper is structured as follows: In Section \ref{sec:def}, we give the basic definitions of second quantization and the ITP framework.
Section \ref{sec:bogoliubovtrafo} recaps known properties on Bogoliubov transformations that are needed for the extended implementation.
In Section \ref{sec:extensionbogoliubov}, we define the $ \cV $-dependent ITP spaces $ \widehat{\sH} $ and prove that creation and annihilation operators are well-defined on them (Lemma \ref{lem:aadaggerexist}).
On these ITP spaces, we define implementability in Section \ref{sec:implementation}, and prove that it is satisfied (Theorems \ref{thm:bosoniccountable} and \ref{thm:fermioniccountable}).
Section \ref{sec:diagonalization}, is devoted to the diagonalization of Hamiltonians in the extended sense (Propositions \ref{prop:diagonalizableboson} and \ref{prop:diagonalizablefermion}).
In Section \ref{sec:applications} we examine two examples for a diagonalization in the extended sense.\\


\section{Basic Definitions}
\label{sec:def}

\subsection{Fock Space Notions}
\label{subsec:fockspace}

We consider a measure space $ (X, \mu) $ with $ X \subseteq \RRR^d $, where we focus on $ X = \RRR^d $ and $ X = \NNN $. A configuration of $ N \in \NNN_0 $ particles is given by the tuple $ q = (\bx_1, \ldots, \bx_N) $, which is an element of the ordered configuration space
\begin{equation}
	\cQ(X) := \bigsqcup_{N = 0}^\infty \cQ(X)^{(N)} := \bigsqcup_{N = 0}^\infty X^N \;.
\label{eq:configdef}
\end{equation}
$ \cQ(X) $ allows for a standard topology and a measure $ \mu_N $ on each sector $ \cQ(X)^{(N)} $, hence yielding a topology and a measure $ \mu_\cQ $ on $ \cQ(X) $. The full Fock space is then
\begin{equation}
	\sF(X) := L^2(\cQ(X),\mu_\cQ) \;.
\label{eq:fockdef}
\end{equation}
The corresponding scalar product is $ \langle \Phi, \Psi \rangle := \int_{\cQ(X)} \overline{\Phi(q)} \Psi(q) \; \di q $ with $ \Phi, \Psi \in \sF(X) $ and the overline denoting complex conjugation. For $ \Psi \in C_0(\cQ(X)) $, a unique continuous representative function exists. This includes smooth functions $ \Psi \in C^\infty(\cQ(X)) $ and smooth functions with compact support $ \Psi \in C_c^\infty(\cQ(X)) $.\\
For describing bosonic/fermionic particle exchange symmetries, we introduce the symmetrization operators $ S_+, S_-: \sF(X) \to \sF(X) $ defined by
\begin{equation}
	(S_\pm \Psi)(\bx_1, \ldots, \bx_N) := \frac{1}{N!} \sum_{\sigma \in S_N} (\pm 1)^{(1-\sgn(\sigma))/2} \Psi(\bx_{\sigma(1)}, \ldots, \bx_{\sigma(N)}) \;,
\label{eq:symmetricdef}
\end{equation}
with permutation group $ S_N $. Here, $ (1-\sgn(\sigma))/2 $ is 0 if the permutation is even, and 1 if it is odd. The bosonic ($ + $) and fermionic ($ - $) Fock space is given by
\begin{equation}
	\sF_\pm(X) := S_\pm [\sF(X)] \;.
\label{eq:symmetricfockdef}
\end{equation}
The $ (N) $-sectors of Fock space are $ \sF(X)^{(N)}:= L^2(\cQ(X)^{(N)}, \CCC) $, with symmetrized and antisymmetrized analogues $ \sF_\pm^{(N)} $. Note that we may equivalently write
\begin{equation}
	\sF(X) := \bigoplus_{N = 0}^\infty \sF(X)^{(N)} \;, \qquad
	\sF(X)^{(N)} := \underbrace{\fh \otimes \ldots \otimes \fh}_{N \text{ times}} \;,
\end{equation}
with $ \fh := L^2(X) $. A change of basis then allows us to identify $ \fh $ with $ \ell^2 = L^2(\NNN) $ (or a subspace thereof), thus identifying $ \sF(X) $ with $ \sF(\NNN) $. In the following, we drop the $ (X) $ if not explicitly needed.\\

Creation and annihilation operators $ a^\dagger(f), a(f) $ for some $ f \in \fh $ can be defined by using $ q \setminus \bx_j \in X^{N-1} $ for denoting the removal of one particle $ \bx_j $ from a configuration $ q \in X^N $:
\begin{equation}
\begin{aligned}
	(a_\pm^\dagger(f) \Psi)(q) &= \sum_{j = 1}^N \frac{(\pm 1)^j}{\sqrt{N}} f(\bx_j) \Psi(q \setminus \bx_j) \;,\\
	(a_\pm(f) \Psi)(q) &= \sqrt{N+1} \int \overline{f(\bx)} \Psi(q,\bx) \; \di \mu(\bx) \;.
\end{aligned}
\label{eq:aadagger}
\end{equation}
It is well-known that the fermionic operators $ a_-, a^\dagger_- $ are bounded and hence defined on all $ \Psi \in \sF_- $, while the bosonic $ a_+, a^\dagger_+ $ are unbounded, but can still be defined on a dense subspace of $ \sF_+ $. Further, \eqref{eq:aadagger} implies the canonical commutation/anticommutation relations (CCR/CAR):
\begin{equation}
	[a_\pm(f),a_\pm^\dagger(g)]_\pm = \langle f,g \rangle_\fh \;, \qquad
	[a_\pm(f),a_\pm(g)]_\pm = 0 = [a_\pm^\dagger(f),a_\pm^\dagger(g)]_\pm \;,
\label{eq:CARCCR}
\end{equation}
with commutator $ [A,B]_+ = [A,B] = AB-BA $ and anticommutator $ [A,B]_- = \{A,B\} = AB+BA $. In the following, we will drop the indices $ \pm $ if there is no risk of confusion.\\
It is also customary to just consider $ a(f), a^\dagger(f) $ not as operators, but as formal expressions within a $ ^* $-algebra\footnote{Here, for $ \cA_- $, the term ``generated'' comprises taking the closure.}
\begin{equation}
	\cA = \cA_\pm \quad \text{generated by} \quad
	\{a_\pm(f), a_\pm^\dagger(f) \; \mid \; f \in \fh \} \;.
\label{eq:cA}
\end{equation}
The involution is given by $ a(f)^* = a^\dagger(f) $ and the multiplication in $ \cA $ is such that the CCR/CAR hold. In particular, $ \cA_- $ is a $ C^*$-algebra by boundedness of operators.\\

\subsection{Infinite Tensor Products}
\label{subsec:infinitetensorprod}

In this subsection, we give a quick introduction to general ITPs, as introduced by von Neumann \cite{vonNeumann1939}, see also \cite{BerezanskyKondratiev}. Some useful lemmas and remarks concerning this construction are given in Appendix \ref{sec:resultsITP}. For a more thorough discussion, we refer the reader to \cite{vonNeumann1939}.\\

We consider a (possibly uncountable) index set $ I $, and for each $ k \in I $ a Hilbert space $ \sH_k $ with scalar product $ \langle \cdot , \cdot \rangle_k $ and induced norm $ \Vert \cdot \Vert_k $. The aim is to construct a vector space, which is generated formally by ITPs
\begin{equation}
	\Psi = \prod_{k \in I}^\otimes \Psi_k \;,
\end{equation}
or equivalently, by families $ (\Psi) = (\Psi_k)_{k \in I} , \; \Psi_k \in \sH_k $. If $ I $ is countable, then $ (\Psi) = (\Psi_1, \Psi_2, \ldots ) $ defines a sequence. Each family gets assigned the formal expression
\begin{equation}
	\Vert (\Psi) \Vert := \prod_{k \in I} \Vert \Psi_k \Vert_k \;,
\label{eq:infinitenorm}
\end{equation}
which we will later use for defining a norm. In order to answer the question, whether \eqref{eq:infinitenorm} defines a complex number, one introduces the notions of convergence within a (possibly uncountable) sum:
\begin{itemize}
\item For $ z_k \in \CCC, \; k \in I $, we call $ \sum_{k \in I} z_k $ or $ \prod_{k \in I} z_k $ convergent to $ a \in \CCC $, if for all $ \delta > 0 $, there exists some finite set $ I_\delta $, such that for all finite sets $ J \subseteq I $ with $ I_\delta \subseteq J $, we have
\begin{equation}
	\left\vert a -\sum_{k \in J} z_k \right\vert \le \delta \qquad 
	\text{or} \qquad
	\left\vert a -\prod_{k \in J} z_k \right\vert \le \delta \;, \qquad \text{respectively} \;.
\end{equation}
\end{itemize}
A simple consequence of this definition is that $ \sum_{k \in I} z_k $ can only converge if $ z_k \neq 0 $ occurs for only countably many $ k \in I $. So the question of convergence reduces to that of sequence convergence. Further, it is shown in \cite{vonNeumann1939} that $ \prod_{k \in I} z_k < \infty $ if and only if we have $ z_k = 0 $ for at least one $ k \in I $ or if $ \sum_{k \in I} \vert z_k - 1 \vert < \infty $. The heuristic reason is that $ \prod_{k \in I} z_k = \exp{\left( \sum_{k \in I} \ln{z_k} \right)} $ and $ \ln{z_k} $ can be linearly approximated near $ 1 $ as $ \ln{z_k} = 1 - z_k + \mathcal{O}((1-z_k)^2) $.\\
In case $ \prod_{k \in I} |z_k| $ converges to a nonzero number, then $ \prod_{k \in I} z_k $ converges if and only if no infinite phase variation occurs. That is, if $ \mathrm{arg}(z_k) \in (-\pi, \pi] $ is the phase of the complex number $ z_k $, then it is required that
\begin{equation}
	\sum_{k \in I} |\mathrm{arg}(z_k)| < \infty \;.
\label{eq:infinitephysecondition}
\end{equation}
In order to establish a notion of convergence, even when \eqref{eq:infinitephysecondition} is violated, one defines that $ \prod_{k \in I} z_k $ is \emph{quasi-convergent} if and only if $ \prod_{k \in I} |z_k| $ converges.
A family $ (\Psi) = (\Psi_k)_{k \in I} $ is now called a
\begin{itemize}
\item $ C ${\bf -sequence} ($ (\Psi) \in \Cseq $) if and only if $ \prod _{k \in I} \Vert \Psi_k \Vert_k < \infty $ ,
\item $ C_0 ${\bf -sequence} if and only if $ \sum _{k \in I} \left\vert \Vert \Psi_k \Vert_k - 1 \right\vert < \infty \quad \Leftrightarrow \quad \sum _{k \in I} \left\vert \Vert \Psi_k \Vert_k^2 - 1 \right\vert < \infty $ .
\end{itemize}
Each $ C_0 $-sequence is also a $ C $-sequence. For all $ C $-sequences, we have a well-defined value $ \Vert (\Psi) \Vert \in \CCC $ by \eqref{eq:infinitenorm} and each $ C $-sequence, that is not a $ C_0 $-sequence, must automatically satisfy $ \Vert (\Psi) \Vert = 0 $. However, \eqref{eq:infinitenorm} only defines a seminorm, since there exist $ (\Psi) \neq 0 $ with $ \prod_{k \in I} \Vert \Psi_k \Vert_k = 0 $. To make it a norm, one defines
\begin{itemize}
\item $ \prod^{\circled{$\otimes$}}_{k \in I} \sH_k $ as the space of all functionals on $ \Cseq $ which are conjugate-linear in each component.
\end{itemize}
Following \cite{vonNeumann1939}, we can embed $ \iota: \Cseq \to \prod^{\circled{$\otimes$}}_{k \in I} \sH_k $ by identifying $ (\Phi) \in \Cseq $ with the functional
\begin{equation}
	\Phi = \iota((\Phi)): (\Psi) \mapsto \prod_{k \in I} \langle \Phi_k, \Psi_k \rangle_k \;.
\end{equation}
This identification essentially sets up an equivalence relation $ \sim_{\rm C} $ on $ \Cseq $, where $ (\Phi) \sim_{\rm C} (\Phi') $, whenever $ \iota((\Phi)) = \iota((\Phi')) $. In Proposition \ref{prop:cksequence}, we show that equivalence is given if and only if $ (\Phi) $ and $ (\Phi') $ just differ by a family of complex factors $ (c_k)_{k \in I} $ with $ \prod_{k \in I} c_k = 1 $. The functionals in $ \iota[\Cseq] $ are then the equivalence classes and the span of these functionals is denoted by \cite{vonNeumann1939}:
\begin{itemize}
\item $ \widetilde{\prod}^{\otimes}_{k \in I} \sH_k := \vecspan(\iota[\Cseq])$.
\end{itemize}
In the following, we may drop the embedding map $ \iota $ and simply identify $ (\Phi) $ with $ \Phi $, whenever the identification is obvious. An inner product $ \langle \cdot , \cdot \rangle $ can uniquely be defined on $ \widetilde{\prod}^{\otimes}_{k \in I} \sH_k $ via
\begin{equation}
	\langle \Phi , \Psi \rangle = \prod_{k \in I} \langle \Phi_k, \Psi_k \rangle_k \;,
\label{eq:infinitescalarproduct}
\end{equation}
which is a quasi-convergent product that we understand to be 0 whenever it is not convergent. $ \langle \cdot , \cdot \rangle $ makes $ \widetilde{\prod}^{\otimes}_{k \in I} \sH_k $ a pre-Hilbert space and induces a norm $ \Vert \Phi \Vert $ agreeing with \eqref{eq:infinitenorm} under identification $ \Vert \Phi \Vert = \Vert (\Phi) \Vert $. This norm allows completing $ \widetilde{\prod}^{\otimes}_{k \in I} \sH_k $ to a Hilbert space:
\begin{itemize}
\item The \textbf{infinite tensor product space} $ \widehat{\sH} = \prod^{\otimes}_{k \in I} \sH_k $ is defined as the space of all $ \Phi \in \prod^{\circled{$\otimes$}}_{k \in I} \sH_k $, such that there exists a Cauchy sequence $ (\Phi^{(r)})_{r \in \NNN} \subset \widetilde{\prod}^{\otimes}_{k \in I} \sH_k $ with respect to $ \Vert \cdot \Vert $ that converges to $ \Phi $ in the weak-$ ^* $ topology on $ \prod^{\circled{$\otimes$}}_{k \in I} \sH_k $.
\end{itemize}
One may indeed extend $ \langle \cdot, \cdot \rangle $ to $ \widehat{\sH} $, making the latter a Hilbert space \cite{vonNeumann1939}. In order to further analyze its structure, we divide $ \widehat{\sH} $ into subspaces, For which we split the set of $ C_0 $-sequences into equivalence classes via
\begin{itemize}
\item equivalence: $ (\Phi) \sim (\Psi) \quad : \Leftrightarrow \quad \sum_{k \in I} \left\vert \langle \Phi_k, \Psi_k \rangle - 1 \right\vert < \infty $
\item weak equivalence: $ (\Phi) \sim_w (\Psi) \quad : \Leftrightarrow \quad \sum_{k \in I} \big\vert |\langle \Phi_k, \Psi_k \rangle| - 1 \big\vert < \infty $
\end{itemize}
The respective equivalence classes are called $ C $ and $ C_w $, and the corresponding linear spaces of an equivalence class are
\begin{itemize}
\item $ \prod^{\otimes C}_{k \in I} \sH_k := \overline{\vecspan\{ \Psi \; \mid \; \exists (\Psi) \in C: \iota((\Psi)) = \Psi \}}^{\Vert \cdot \Vert} $ for equivalence
\item $ \prod^{\otimes C_w}_{k \in I} \sH_k := \overline{\vecspan\{ \Psi \; \mid \; \exists (\Psi) \in C_w: \iota((\Psi)) = \Psi \}}^{\Vert \cdot \Vert} $ for weak equivalence
\end{itemize}

Now, each $ C_0 $-sequence $ (\Psi) $ in some equivalence class $ [(\Omega)] = C $ (with $ (\Omega) \in \Cseq $ being interpreted as the vacuum vector) can be written in \textbf{coordinates} as follows \cite[Thm. V]{vonNeumann1939}: Choose an orthonormal basis $ (e_{k,n})_{n \in \NNN_0} $ for each $ \sH_k $, such that $ \Omega_k = e_{k,0} $ (we think of $ e_{k,0} $ as mode $ k $ being in the vacuum). Then, $ (\Psi) = (\Psi_k)_{k \in I} $ is uniquely specified by the coordinates $ c_{k,n} := \langle e_{k,n}, \Psi_k \rangle_k \in \CCC $. In this coordinate representation, it is true that
\begin{itemize}
\item $ \prod^{\otimes C}_{k \in I} \sH_k $ is the closure of the space spanned by all normalized $ C_0 $-sequences, where $ c_{k,0} = 1 $ for all but finitely many $ k \in I $.
\end{itemize}
Or heuristically speaking, ``almost all $ \Psi_k $ are in the vacuum''. By \cite[Thm. V]{vonNeumann1939}, also a generic $ \Psi \in \prod^{\otimes C}_{k \in I} \sH_k $ can be written as
\begin{equation}
	\Psi = \sum_{n(\cdot) \in F} a(n(\cdot)) \prod^\otimes_{k \in I} e_{k,n(k)} \;,
\label{eq:Phiincoordinates}
\end{equation}
with $ F $ being the countable set of all functions $ n:I \to \NNN_0 $ with $ n(k) = 0 $ for almost all $ k \in I $, and $ a(n(\cdot)) \in \CCC $ being the coordinates of $ \Psi $ with $ \sum_{n(\cdot) \in F} |a(n(\cdot))|^2 < \infty $.\\

\section{Bogoliubov Transformations}
\label{sec:bogoliubovtrafo}

In this section, we introduce our notation for Bogoliubov transformations and recap some important properties. Standard references on the subject are \cite{Solovej2014} and \cite{BachLiebSolovej1994}.\\

\subsection{Transformation on Operators}
\label{subsec:transformationsonoperators}

Consider the one-operator subspace $ W_1 $ of $ \cA $, which is linearly spanned by $ a_\pm^\dagger(f), a_\pm(g) $, $ f, g \in \fh $. By an algebraic Bogoliubov transformation, we mean any bijective map $ \cV_A: W_1 \to W_1 $, which sends $ a_\pm^\dagger(f), a_\pm(g) $ to a new set of creation and annihilation operators $ b_\pm^\dagger(f), b_\pm(g) $, such that $ b_\pm^\dagger(f) $ is the adjoint of $ b_\pm(f) $ and the CAR/CCR are conserved under $ \cV_\cA $ and its adjoint.\\
To facilitate the presentation, we fix a basis $ (e_j)_{j \in \NNN} \subset \fh $ which identifies $ f \in \fh $ with an equally denoted vector $ \bof = (f_j)_{j \in \NNN} \in \ell^2 $ by $ f_j := \langle e_j, f \rangle $. This way, we can also identify $ a^\dagger(f) + a(\overline{g}) \in W_1 $ with a vector $ (\bof, \bg) \in \ell^2 \oplus \ell^2 $, so an algebraic Bogoliubov transformation $ \cV_A $ is identified with a linear operator $ \cV $ on $ \ell^2 \oplus \ell^2 $, which we just call ``Bogoliubov transformation''.
We may also encode sums of creation and annihilation operators by vector pairs $ (\bof,\bg) \in \ell^2 \oplus \ell^2 $ via the generalized creation/annihilation operators
\begin{equation}
\begin{aligned}
	A_\pm^\dagger: \ell^2 \oplus \ell^2 &\to \cA_\pm \;,
	&(\bof_1, \bof_2) &\mapsto a_\pm^\dagger(\bof_1) + a_\pm(\overline{\bof_2}) = \sum_j (f_{1,j} a_\pm^\dagger(e_j) + f_{2,j} a_\pm(e_j)) \;,\\
	A_\pm: \ell^2 \oplus \ell^2 &\to \cA_\pm \;,
	&(\bg_1, \bg_2) &\mapsto a_\pm(\bg_1) + a_\pm^\dagger(\overline{\bg_2}) = \sum_j (\overline{g_{1,j}} a_\pm(e_j) + \overline{g_{2,j}} a_\pm^\dagger(e_j)) \;.\\
\label{eq:genAAdagger}
\end{aligned}
\end{equation}
A Bogoliubov transformation is then encoded by a $ 2 \times 2 $ block matrix
\begin{equation}
	\cV = \begin{pmatrix} u & v \\ \overline{v} & \overline{u} \end{pmatrix} \;,
\label{eq:cV}
\end{equation}
with operators $ u,v: \ell^2 \to \ell^2 $. The case of unbounded $ u,v $ is treated later in Section \ref{sec:extensionbogoliubov}. The Bogoliubov transformed operators are then given by 
\begin{equation}
\begin{aligned}
	b_\pm^\dagger(\bof) &= A_\pm^\dagger(\cV (\bof,0)) = a_\pm^\dagger(u \bof) + a_\pm(v \overline{\bof}) \;,\\
	b_\pm(\bg) &= A_\pm(\cV (\bg,0)) =a_\pm^\dagger(v \overline{\bg}) + a_\pm(u \bg) \;.\\
\end{aligned}
\label{eq:BogoliubovTransformation}
\end{equation}

In order for $ \cV $ to be a Bogoliubov transformation, we require that both $ \cV $ and $ \cV^* $ conserve the CAR/CCR, so
\begin{equation}
	[b_\pm(\bof),b_\pm^\dagger(\bg)]_\pm = \langle \bof,\bg \rangle \;, \qquad
	[b_\pm(\bof),b_\pm(\bg)]_\pm = 0 = [b_\pm^\dagger(\bof),b_\pm^\dagger(\bg)]_\pm \;,
\label{eq:CARCCRb}
\end{equation}
and the same, if in \eqref{eq:BogoliubovTransformation} $ \cV $ is replaced by $ \cV^* $. An explicit calculation shows that this conservation is equivalent to the Bogoliubov relations
\begin{equation}
\begin{aligned}
	u^*u \mp v^T \overline{v} &= 1 \;, \qquad
	&u^*v \mp v^T \overline{u} = 0 \;,\\
	u u^* \mp v v^* &= 1 \;, \qquad
	&u v^T \mp v u^T =0 \;,\\
\end{aligned}
\label{eq:Bogoliubovrelations}
\end{equation}
with $ (u^T)_{ij} = u_{ji}, \; (\overline{u})_{ij} = \overline{u_{ij}}, \; (u^*)_{ij} = \overline{u_{ji}} $ and the same for $ v_{ij} $. Since $ v^*v $ is self-adjoint, we have $ v^*v = \overline{v^*v} = v^T \overline{v} $, so the first equation is equivalent to $ u^*u \mp v^*v = 1 $. The generalized creation and annihilation operators also allow for particularly easy ``generalized CAR/CCR'': Using the standard scalar product on $ \boF,\bG \in \ell^2 \oplus \ell^2 $:
\begin{equation}
	\langle \boF, \bG \rangle
	= \Big\langle \begin{pmatrix} \bof_1 \\ \bof_2 \end{pmatrix} , \begin{pmatrix} \bg_1 \\ \bg_2 \end{pmatrix} \Big\rangle 
	= \sum_j ( \overline{f_{1,j}} g_{1,j} + \overline{f_{2,j}} g_{2,j}) \;,
\end{equation}
and $ \cS_- = \mathrm{id} $, $ \cS_+ = \left( \begin{smallmatrix} 1&0 \\ 0&-1 \end{smallmatrix} \right) $, we obtain the generalized CAR/CCR:
\begin{equation}
	\; [A_\pm(\boF), A_\pm^\dagger(\bG)]_\pm = \langle \boF, \cS_\pm \bG \rangle \;, \qquad
	[A_\pm(\boF), A_\pm(\bG)]_\pm = [A_\pm^\dagger(\boF), A_\pm^\dagger(\bG)]_\pm = 0 \;.
\end{equation}

\subsection{Implementation on Fock Space}
\label{subsec:implementation1}

In the above notation, a Bogoliubov transformation is implementable (in the regular sense), if there exists a unitary operator $ \UUU_\cV : \sF \to \sF $, such that
\begin{equation}
	\UUU_\cV A^\dagger(\boF) \UUU_\cV^* = A^\dagger(\cV \boF) \;.
\end{equation}
We recap some of the basic steps of the implementation process presented in \cite{Solovej2014} for $ \tr(v^* v) < \infty $, as they have to be carried out in a slightly modified way for an implementation in the extended sense.\\
The main task within the implementation is to find the \emph{Bogoliubov vacuum} $ \Omega_\cV \in \sF $, which is the vector annihilated by all operators $ b(\bof) $:
\begin{equation}
	b(\bof) \Omega_\cV = (a^\dagger(u \bof) + a(v\overline{\bof})) \Omega_\cV = 0 \qquad \forall \bof \in \ell^2 \;.
\end{equation}
If we can find such an $ \Omega_\cV $, then it is an easy task to transform any product state vector $ a^\dagger(\bof_1) \ldots a^\dagger(\bof_n) \Omega \in \sF $ via:
\begin{equation}
	\UUU_\cV a^\dagger(\bof_1) \ldots a^\dagger(\bof_n) \Omega
	= b^\dagger(\bof_1) \ldots b^\dagger(\bof_n) \Omega_\cV \;.
\label{eq:vacuumtrafo}
\end{equation}
The span of these state vectors (also called algebraic tensor product) is dense within $ \sF $, so we can transform any $ \Psi \in \sF $ by means of \eqref{eq:vacuumtrafo}.\\
To find $ \Omega_\cV $, we decompose the transformation $ \cV $ into modes by constructing vectors $ \bof_j $, such that $ u \bof_j, v \overline{\bof_j} $ are proportional to the same normalized vector $ \bg_j \in \ell^2 $, i.e.
\begin{equation}
	\cV \begin{pmatrix} \bof_j\\0 \end{pmatrix}
	= \begin{pmatrix} u \bof_j\\ \overline{v} \bof_j \end{pmatrix}
	= \begin{pmatrix} \mu_j \bg_j\\ \overline{\nu_j \bg_j} \end{pmatrix} \qquad \Leftrightarrow \qquad \UUU_\cV \; a^\dagger(\bof_j) \; \UUU_\cV^* = \mu_j a^\dagger(\bg_j) + \nu_j a(\bg_j) \;,
\label{eq:fg}
\end{equation}
where $ \mu_j, \nu_j \in \CCC $. If \eqref{eq:fg} holds, we only have to solve $ (\nu_j a^\dagger(\bg_j) + \mu_j a(\bg_j) )\Omega_\cV $ for each $ \bg_j $, separately.\\

\subsubsection{Bosonic Case}
\label{subsubsec:bosoniccaseimplementation}

Here, \eqref{eq:fg} can indeed be fulfilled: Following \cite{Solovej2014}, introducing the complex conjugation $ J \tilde{\bh}_j = J^* \tilde{\bh}_j = \overline{\tilde{\bh}_j} $, the operator $ C := u^* v J $ has an eigenbasis $ (\bof_j) $ with $ C \bof_j = \lambda_j \bof_j, \; \lambda_j \in \CCC $. One may further show that $ \bg_j := \mu_j^{-1} u \bof_j $ defines an orthonormal basis and that $ v \overline{\bof_j} = v J \bof_f = \nu_j \bg_j $ where $ \mu_j, \nu_j \ge 0 $ are fixed by
\begin{equation}
	\mu_j^2 - \nu_j^2 = 1 \;, \qquad
	\lambda_j = \mu_j \nu_j \;.
\end{equation}

The condition $ (\mu_j a(\bg_j) + \nu_j a^\dagger(\bg_j)) \Omega_\cV = 0 $ leads to one recursion relation per mode, which is formally solved by 
\begin{equation}
	\Omega_\cV = \left( \prod_j \left(1 - \tfrac{\nu_j^2}{\mu_j^2} \right)^{1/4} \right) \exp{\left( -\sum_j \frac{\nu_j}{2 \mu_j} (a^\dagger(\bg_j))^2 \right) } \Omega \;.
\label{eq:bosonicbogoliubovvacuum}
\end{equation}
The Shale condition now indicates when $ \Omega_\cV $ lies in Fock space, which is if and only if $ \sum_j \frac{\nu_j^2}{\mu_j^2} < \infty \Leftrightarrow \sum_j \nu_j^2 < \infty $. In that case, the transformation is implemented by \cite[(3.1)]{Schlein2019}:
\begin{equation}
	\UUU_\cV = \exp{ \left(- \sum_j \tfrac{\xi_j}{2} ((a^\dagger(\bg_j))^2 - (a(\bg_j))^2 ) \right) } \UUU_{\bg \bof} =: \prod_{j \in \NNN} \UUU_{j,\cV} \;,
\label{eq:bosonicbogoliubovimplementer}
\end{equation}
\begin{equation}
	\text{with} \qquad \sinh \xi_j := \nu_j \qquad \Rightarrow \qquad \cosh \xi_j := \mu_j \;,
\end{equation}
and where $ \UUU_{\bg \bof}: \sF \to \sF $ is the unitary basis change transformation
\begin{equation}
	\UUU_{\bg \bof}: \bof_{j_1} \otimes \ldots \otimes \bof_{j_n} \mapsto \bg_{j_1} \otimes \ldots \otimes \bg_{j_n} \qquad \forall \; j_1, \ldots, j_n \in \NNN \;.
\label{eq:Ugf}
\end{equation}
For a more general discussion on $ \UUU_\cV $, we refer the reader to \cite[Thm. 16.47]{DerezinskiGerard2013}.\\

\subsubsection{Fermionic Case}
\label{subsubsec:fermioniccaseimplementation}

In the fermionic case, following \cite{Solovej2014}, we can find a common orthonormal eigenbasis $ (\bof_j)_{j \in J} $ of $ C^*C $ (eigenvalues $ \lambda_j^2 $), of $ u^*u $ (eigenvalues $ \mu_j^2 $), and of $ v^*v $ (eigenvalues $ \nu_j^2 = 1 - \mu_j^2 $). The index set can be split as $ J = J' \cap J'' \subseteq \NNN $ with $ J'' $ containing all indices of zero eigenvectors and $ J' $ those of nonzero eigenvectors. Further, the $ j \in J' $ can be rearranged in pairs as $ J' := \{j \; \mid \; j=2i \lor j=2i-1, \; i \in I'\} $ such that
\begin{equation}
	C \bof_{2i} = \lambda_{2i} \bof_{2i-1} \;, \qquad
	C \bof_{2i-1} = - \lambda_{2i} \bof_{2i} \;,
\label{eq:fermionicpairs}
\end{equation}
for some $ I' \subseteq \NNN $. For $ j \in J'' $, two cases may occur: If $ \nu_j = 1 $, then we have a particle--hole transformation, for which we write $ j \in J''_1 $ and get
\begin{equation}
	\cV \begin{pmatrix} \bof_j \\ 0 \end{pmatrix} = \begin{pmatrix} 0 \\ \overline{\boeta_j} \end{pmatrix} \;, \qquad
	j \in J''_1 \;,
\label{eq:fermionictrafo1}
\end{equation}
for a suitable choice of the phase $ \boeta_j = e^{i \varphi} v \overline{\bof_j} $. The case $ \nu_j = 0 $ will be denoted by $ j \in J''_0 = J'' \setminus J''_1 $, and we have
\begin{equation}
	\cV \begin{pmatrix} \bof_j \\ 0\end{pmatrix} = \begin{pmatrix} \boeta_j \\ 0 \end{pmatrix} \;, \qquad
	j \in J''_0 \;,
\label{eq:fermionictrafo2}
\end{equation}
for a suitable phase choice of $ \boeta_j = e^{i \varphi} u \bof_j $. In case $ j \in J' $ with $ i \in I' \Rightarrow \lambda_{2i} \neq 0, \mu_{2i} \neq 0 $ (Cooper pair), we may define the normalized vectors
\begin{equation}
	\boeta_{2i} := \alpha_i^{-1} \; u \bof_{2i} \;, \qquad 
	\boeta_{2i-1} := \alpha_i^{-1} \; u \bof_{2i-1} \;,
\end{equation}
where $ \alpha_i, \beta_i > 0 $, $ \alpha_i = \mu_{2i} = \mu_{2i-1} $, $ \alpha_i^2 + \beta_i^2 = 1 $, and for which
\begin{equation}
	\cV \begin{pmatrix} \bof_{2i} \\ 0\end{pmatrix} = \begin{pmatrix} \alpha_i \; \boeta_{2i} \\ \overline{\beta_i \; \boeta_{2i-1}} \end{pmatrix} \;, \qquad
	\cV \begin{pmatrix} \bof_{2i-1} \\ 0\end{pmatrix} = \begin{pmatrix} \alpha_i \; \boeta_{2i-1} \\ \overline{-\beta_i \; \boeta_{2i}} \end{pmatrix} \;, \qquad i \in I' \;.
\label{eq:fermionictrafo3} 
\end{equation}
Relations \eqref{eq:fermionictrafo1}, \eqref{eq:fermionictrafo2} and \eqref{eq:fermionictrafo3} now replace \eqref{eq:fg}. The Bogoliubov vacuum is
\begin{equation}
	\Omega_\cV = \left( \prod_{j \in J''_1} a^\dagger(\boeta_j) \right) \left( \prod_{i \in I'} (\alpha_i - \beta_i a^\dagger(\boeta_{2i}) a^\dagger(\boeta_{2i-1}) ) \right) \Omega \;,
\label{eq:fermionifbogoliubovvacuum}
\end{equation}
and the implementer is
\begin{equation}
\begin{aligned}
	\UUU_\cV &= \left( \prod_{j \in J''_1} (a^\dagger(\boeta_j) + a(\boeta_j)) \right) \exp \left( -\sum_{i \in I'} \xi_i ( a^\dagger(\boeta_{2i}) a^\dagger(\boeta_{2i-1}) - a(\boeta_{2i-1}) a(\boeta_{2i}) ) \right) \UUU_{\boeta \bof}\\
	\UUU_\cV &=: \left(\prod_{j \in J''} \UUU_{j,\cV}\right) \left( \prod_{i \in I'} \UUU_{2i,2i-1,\cV} \right) \;, \quad
	\text{with} \quad \sin \xi_i := \beta_i \quad \Rightarrow \quad \cos \xi_i := \alpha_i \;,
\end{aligned}
\label{eq:fermionicbogoliubovimplementer}
\end{equation}
and where $ \UUU_{\boeta \bof} $ is the unitary basis change transformation
\begin{equation}
	\UUU_{\boeta \bof}: \bof_{j_1} \otimes \ldots \otimes \bof_{j_n} \mapsto \boeta_{j_1} \otimes \ldots \otimes \boeta_{j_n} \qquad \forall \; j_1, \ldots, j_n \in J \;.
\label{eq:Uetaf}
\end{equation}

\section{Bogoliubov Transformations: Extended}
\label{sec:extensionbogoliubov}

In the extended case, $ v $ is possibly unbounded, so we must show that the Bogoliubov relations \eqref{eq:Bogoliubovrelations} survive the extension. We do this in Section \ref{subsec:extensionbogoliubov}, while preparing the spectral decomposition for the final ITP construction. In Section \ref{subsec:extopalg}, we define an extended $ ^* $-algebra $ \cAB_{\be} $ of creation and annihilation operator products, and in Section \ref{subsec:finalITP}, we finish the construction of $ \widehat{\sH} $, with respect to a given Bogoliubov transformation $ \cV $.\\

\subsection{Extension of the Bogoliubov Relations}
\label{subsec:extensionbogoliubov}

	Throughout the following construction, we will assume that $ v^*v $ is densely defined and self-adjoint. In that case, we can define the self-adjoint operators $ C^*C := v^*v (1 \pm v^*v) $ and $ |C| = \sqrt{C^*C} $ by spectral calculus. By the spectral theorem in the form of \cite[Thm. 10.9]{Hall2013}, we may then decompose $ \ell^2 $ as a direct integral
\begin{equation}
	\ell^2 = \int^\oplus_{\sigma(|C|)} \CCC^n \; \di \mu_1(\lambda) \;,
\label{eq:directintegral}
\end{equation}
where $ \sigma(|C|) = \sigma $ is the spectrum of $ |C| $, $ \mu_1 $ is a suitable measure on it and $ n: \sigma \to \NNN \cup \{\infty\} $ \cite[Def. 7.18]{Hall2013} is a measurable dimension function. Put differently, as visualized in Figure \ref{fig:X_spectrum}, we can find a spectral set $ X = \bigcup_{\lambda \in \sigma} \{\lambda\} \times Y_\lambda \subset \RRR^2 $ with $ Y_\lambda \subseteq \ZZZ, |Y_\lambda| = n(\lambda) $ accounting for multiplicity and unitary maps\footnote{See also the formulation \cite[Thm. 10.10]{Hall2013} of the spectral theorem.}
\begin{equation}
	U_{X \bof}: \ell^2 \to L^2(X) \;, \qquad
	U_{\bof X} = U_{X \bof}^{-1} \;,
\label{eq:Ux}
\end{equation}
such that
\begin{equation}
	|C| = U_{\bof X} \lambda U_{X \bof} \;,
\end{equation}
with $ \lambda $ being the operator on $ L^2(X) $ that multiplies by $ \lambda(x) $. In addition, we denote $ Y = \bigcup_{\lambda \in \sigma} Y_\lambda \subseteq \ZZZ $ with $ |Y| $ being an upper bound for the multiplicity of any eigenvalue. Note that the $ \lambda $ here correspond to the $ \lambda_j $ in Section \ref{subsec:implementation1}.\\ 
	We also make use of the formulation \cite[Thm. 10.4]{Hall2013} of the spectral theorem, which provides us with a projection-valued measure $ P_{|C|} $, such that
\begin{equation}
	|C| = \int_X \lambda(x) \; dP_{|C|}(x) = \int_{\sigma \times Y} \lambda \; \di P_{|C|}(\lambda, y) \;.
\label{eq:muCC}
\end{equation}
Here, we choose $ Y \subseteq \ZZZ $, so $ X \subset \RRR^2 $ consists of ``lines'' with distance 1.
\begin{figure}
	\centering
	\scalebox{1.0}{\begin{tikzpicture}

\draw[thick, ->] (-2,0) -- (4.1,0) node[anchor = south west]{$ \lambda $};
\draw[thick, ->] (-1.5,-0.8) -- (-1.5,1.8) node[anchor = south west]{$ y $};
\draw (-1.6,-0.5) -- (-1.4,-0.5);
\draw (-1.6,0.5) -- (-1.4,0.5);
\draw (-1.6,1) -- (-1.4,1);
\draw (-1.6,1.5) -- (-1.4,1.5);
\node at (-1.8,-0.5) {-1};
\node at (-1.8,0) {0};
\node at (-1.8,0.5) {1};
\node at (-1.8,1) {2};
\node at (-1.8,1.5) {3};
\node[red] at (4.2,0.8) {$X$};

\fill[red] (-1.1,-1.7) circle (0.07);
\fill[red] (-0.7,-1.7) circle (0.07);
\fill[red] (-0.3,-1.7) circle (0.07);
\draw[red, line width = 3] (0.2,-1.7) -- (1,-1.7);
\fill[red] (1.5,-1.7) circle (0.07);
\draw[red, line width = 3] (2,-1.7) -- (4,-1.7);
\draw[red, rounded corners = 5] (-1.3,-1.9) rectangle (4.2,-1.5) ;
\draw[red] (3.1,-1.5) -- (2.6,-1.3) node[anchor = east]{$ \sigma $};

\fill[gray!50!white, opacity = .6] (-1.1,-0.5) circle (0.05);
\fill[gray!50!white, opacity = .6] (-0.7,-0.5) circle (0.05);
\fill[gray!50!white, opacity = .6] (-0.3,-0.5) circle (0.05);
\draw[gray!50!white, opacity = .6, line width = 2] (0.2,-0.5) -- (1,-0.5);
\fill[gray!50!white, opacity = .6] (1.5,-0.5) circle (0.05);
\draw[gray!50!white, opacity = .6, line width = 2] (2,-0.5) -- (4,-0.5);

\fill[red] (-1.1,0) circle (0.05);
\fill[red] (-0.7,0) circle (0.05);
\fill[red] (-0.3,0) circle (0.05);
\draw[red, line width = 2] (0.2,0) -- (1,0);
\fill[red] (1.5,0) circle (0.05);
\draw[red, line width = 2] (2,0) -- (4,0);

\fill[red] (-1.1,0.5) circle (0.05);
\fill[gray!30!white] (-0.7,0.5) circle (0.05);
\fill[red] (-0.3,0.5) circle (0.05);
\draw[red, line width = 2] (0.2,0.5) -- (1,0.5);
\fill[gray!30!white] (1.5,0.5) circle (0.05);
\draw[red, line width = 2] (2,0.5) -- (4,0.5);

\fill[red] (-1.1,1) circle (0.05);
\fill[gray!30!white] (-0.7,1) circle (0.05);
\fill[red] (-0.3,1) circle (0.05);
\draw[red, line width = 2] (0.2,1) -- (1,1);
\fill[gray!30!white] (1.5,1) circle (0.05);
\draw[gray!30!white, line width = 2] (2,1) -- (4,1);

\fill[gray!30!white] (-1.1,1.5) circle (0.05);
\fill[gray!30!white] (-0.7,1.5) circle (0.05);
\fill[red] (-0.3,1.5) circle (0.05);
\draw[gray!30!white, line width = 2] (0.2,1.5) -- (1,1.5);
\fill[gray!30!white] (1.5,1.5) circle (0.05);
\draw[gray!30!white, line width = 2] (2,1.5) -- (4,1.5);

\end{tikzpicture}}
	\caption{The spectral set $ X $ for a generic spectrum of $ |C| $.}
	\label{fig:X_spectrum}
\end{figure}
The case $ \lambda = 0 $ will turn out to be critical, whence we define
\begin{equation}
	X_{\crit} := \{ x = (\lambda,y) \in X \; \mid \; \lambda = 0\} \;, \qquad
	X_{\reg} := X \setminus X_{\crit} \;.
\label{eq:XcritXreg}
\end{equation}
Our (dense) space of test functions on the spectral set is given by:
\begin{equation}
	\cD_X := C_c^\infty(X_{\crit}) \otimes C_c^\infty(X_{\reg}) \;.
\label{eq:cDX}
\end{equation}
The corresponding test function space in $ \ell^2 $ is
\begin{equation}
	\cD_{|C|} := U_{\bof X} \cD_X \;.
\label{eq:cDC}
\end{equation}
For non-open $ X $, we interpret \eqref{eq:cDX} in the same way as the definition of $ \cE(X) $ \eqref{eq:cEej}:
\begin{equation*}
	C_c^\infty(X) := C_c^\infty(\RRR^2) / \{\phi \; \mid \; \phi(x) = 0 \; \forall x \in X \} \;.
\end{equation*}

\begin{lemma}[Bogoliubov relations \eqref{eq:Bogoliubovrelations} survive the extension]\ \\
Suppose, $ u $ and $ v $ are defined on a common dense domain $ \cD \subseteq \ell^2 $, such that $ v^*v $ is densely defined and self-adjoint, and such that the linear operator
\begin{equation}
	\cV = \begin{pmatrix} u & v \\ \overline{v} & \overline{u} \end{pmatrix} \;, \qquad
	\cV: \cD \oplus \cD \to \ell^2 \oplus \ell^2 \;,
\end{equation}
defines a Bogoliubov transformation. That means, both $ \cV $ and $ \cV^* = \left( \begin{smallmatrix} u^* & v^T \\ v^* & u^T \end{smallmatrix} \right) $ preserve the CAR/CCR, see \eqref{eq:genAAdagger}, \eqref{eq:CARCCRb}.\\
Then $ u, v, \overline{u}, \overline{v}, u^*, v^*, u^T $ and $ v^T $ are well-defined on all of $ \cD_{|C|} $, constructed above \eqref{eq:cDC}. Further, the Bogoliubov relations \eqref{eq:Bogoliubovrelations} hold as a weak operator identity on $ \cD_{|C|} $. Conversely, the extended Bogoliubov relations imply conservation of the CAR/CCR under both $ \cV $ and $ \cV^* $.\\

\label{lem:bogoliubovrelationssurvive}
\end{lemma}

\begin{proof}
	Well-definedness of $ u,v $ on $ \cD_{|C|} $ follows from the polar decompositions
\begin{equation}
	v = U_v |v| \;, \qquad
	u = U_u |u| \;,
\end{equation}
	with unitary $ U_v, U_u: L^2(X) \to \ell^2 $. The operators $ |v| = \sqrt{v^*v} $ and $ |u| = \sqrt{u^*u} = \sqrt{1 \pm v^*v} $ on $ L^2(X) $ are bounded in the fermionic case ($ - $), so $ u $ and $ v $ are defined on all of $ \ell^2 $. In the bosonic case ($ + $), they are spectral multiplications by
\begin{equation}
	\nu(\lambda) = \sqrt{-\frac{1}{2} + \sqrt{\frac{1}{4} + \lambda^2}} \quad \text{and} \quad
	\mu(\lambda) = \sqrt{\frac{1}{2} + \sqrt{\frac{1}{4} + \lambda^2}} \;,
\end{equation}	
which are bounded on each bounded interval in $ \lambda $. Therefore, $ |v| $ and $ |u| $ map $ \cD_X $ into itself, and by definition \eqref{eq:cDC} of $ \cD_{|C|} $, the operators $ v $ and $ u $ map $ \cD_{|C|} \to \ell^2 $.
Well-definedness of $ \overline{u} $ and $ \overline{v} $ on $ \cD_{|C|} $ follows analogously and the domains of the adjoints $ u^*, v^*, u^T $ and $ v^T $ contain the domain of the respective original operators, so they all contain $ \cD_{|C|} $.
The CAR/CCR conservation then follows by a direct computation. Checking that \eqref{eq:Bogoliubovrelations} indeed holds as weak operator identities on $ \cD_{|C|} $ is straightforward to check.
\end{proof}

\begin{figure}
	\centering
	\scalebox{1.0}{\begin{tikzpicture}

\fill[gray!10!white] (-1.5,-0.8) rectangle (4.1,1.8);
\node at (3.9,2) {bosonic};

\draw[thick, ->] (-2,0) -- (4.1,0) node[anchor = south west]{$ \lambda $};
\draw[thick, ->] (-1.5,-0.8) -- (-1.5,1.8) node[anchor = south west]{$ y $};
\draw (-1.6,-0.5) -- (-1.4,-0.5);
\draw (-1.6,0.5) -- (-1.4,0.5);
\draw (-1.6,1) -- (-1.4,1);
\draw (-1.6,1.5) -- (-1.4,1.5);
\node at (-1.8,-0.5) {-1};
\node at (-1.8,0) {0};
\node at (-1.8,0.5) {1};
\node at (-1.8,1) {2};
\node at (-1.8,1.5) {3};
\node[red] at (4.2,0.8) {$X$};

\draw[red, rounded corners = 3] (-1.6,-0.2) rectangle (-1.4,1.2) ;
\draw[red] (-1.5,1.2) -- ++(0.5,0.5) node[anchor = south west]{$ X_{\crit} $};

\fill[gray!50!white, opacity = .6] (-1.5,-0.5) circle (0.05);
\fill[gray!50!white, opacity = .6] (-1,-0.5) circle (0.05);
\fill[gray!50!white, opacity = .6] (-0.5,-0.5) circle (0.05);
\fill[red] (0.5,-0.5) circle (0.05);
\fill[gray!50!white, opacity = .6] (1.5,-0.5) circle (0.05);
\fill[gray!50!white, opacity = .6] (3,-0.5) circle (0.05);
\fill[gray!50!white, opacity = .6] (3.5,-0.5) circle (0.05);

\fill[red] (-1.5,0) circle (0.05);
\fill[red] (-1,0) circle (0.05);
\fill[red] (-0.5,0) circle (0.05);
\fill[red] (0.5,0) circle (0.05);
\fill[red] (1.5,0) circle (0.05);
\fill[red] (3,0) circle (0.05);
\fill[red] (3.5,0) circle (0.05);

\fill[red] (-1.5,0.5) circle (0.05);
\fill[gray!30!white] (-1,0.5) circle (0.05);
\fill[red] (-0.5,0.5) circle (0.05);
\fill[red] (0.5,0.5) circle (0.05);
\fill[red] (1.5,0.5) circle (0.05);
\fill[red] (3,0.5) circle (0.05);
\fill[gray!50!white, opacity = .6] (3.5,0.5) circle (0.05);

\fill[red] (-1.5,1) circle (0.05);
\fill[gray!30!white] (-1,1) circle (0.05);
\fill[red] (-0.5,1) circle (0.05);
\fill[red] (0.5,1) circle (0.05);
\fill[gray!30!white] (1.5,1) circle (0.05);
\fill[gray!50!white, opacity = .6] (3,1) circle (0.05);
\fill[gray!50!white, opacity = .6] (3.5,1) circle (0.05);

\fill[gray!30!white] (-1.5,1.5) circle (0.05);
\fill[gray!30!white] (-1,1.5) circle (0.05);
\fill[red] (-0.5,1.5) circle (0.05);
\fill[red] (0.5,1.5) circle (0.05);
\fill[gray!30!white] (1.5,1.5) circle (0.05);
\fill[gray!50!white, opacity = .6] (3,1.5) circle (0.05);
\fill[gray!50!white, opacity = .6] (3.5,1.5) circle (0.05);

\end{tikzpicture}}
	\hfill
	\scalebox{1.0}{\begin{tikzpicture}

\fill[gray!10!white] (-1.5,-0.8) rectangle (1,1.8);
\draw[thick] (1,-0.1) -- (1,0.1);
\node at (1.1,-0.3) {$ \tfrac 12 $};
\draw[dashed] (1,-0.8) -- (1,1.8);
\node at (2,2) {fermionic};
\node at (2.2,1) {$ 0 \le |C| \le \tfrac 12 $};

\draw[thick, ->] (-2,0) -- (1.5,0) node[anchor = south west]{$ \lambda $};
\draw[thick, ->] (-1.5,-0.8) -- (-1.5,1.8) node[anchor = south west]{$ y $};
\draw (-1.6,-0.5) -- (-1.4,-0.5);
\draw (-1.6,0.5) -- (-1.4,0.5);
\draw (-1.6,1) -- (-1.4,1);
\draw (-1.6,1.5) -- (-1.4,1.5);
\node at (-1.8,-0.5) {-1};
\node at (-1.8,0) {0};
\node at (-1.8,0.5) {1};
\node at (-1.8,1) {2};
\node at (-1.8,1.5) {3};

\fill[gray!30!white] (-1.5,-0.5) circle (0.05);
\fill[red] (-1,-0.5) circle (0.05);
\fill[red] (-0.5,-0.5) circle (0.05);
\fill[gray!30!white] (0.3,-0.5) circle (0.05);
\fill[gray!30!white] (0.7,-0.5) circle (0.05);

\fill[red] (-1.5,0) circle (0.05);
\fill[red] (-1,0) circle (0.05);
\fill[red] (-0.5,0) circle (0.05);
\fill[red] (0.3,0) circle (0.05);
\fill[red] (0.7,0) circle (0.05);

\fill[red] (-1.5,0.5) circle (0.05);
\fill[red] (-1,0.5) circle (0.05);
\fill[red] (-0.5,0.5) circle (0.05);
\fill[red] (0.3,0.5) circle (0.05);
\fill[gray!30!white] (0.7,0.5) circle (0.05);

\fill[red] (-1.5,1) circle (0.05);
\fill[red] (-1,1) circle (0.05);
\fill[gray!30!white] (-0.5,1) circle (0.05);
\fill[gray!30!white] (0.3,1) circle (0.05);
\fill[gray!30!white] (0.7,1) circle (0.05);

\fill[gray!30!white] (-1.5,1.5) circle (0.05);
\fill[red] (-1,1.5) circle (0.05);
\fill[gray!30!white] (-0.5,1.5) circle (0.05);
\fill[gray!30!white] (0.3,1.5) circle (0.05);
\fill[gray!30!white] (0.7,1.5) circle (0.05);

\draw[red, rounded corners = 3] (-1.6,-0.2) rectangle (-1.4,1.2) ;
\draw[red] (-1.5,1.2) -- ++(0.5,0.5) node[anchor = south west]{$ X_{\crit} $};

\end{tikzpicture}}
	\caption{Left: Discrete spectrum of $ |C| $ in the bosonic case.\\ Right: In the fermionic case, $ |C| \le \tfrac 12 $ holds.}
	\label{fig:X_spectrum_countable}
\end{figure}

\begin{lemma}
Let $ v^*v $ be self-adjoint. Then $ C = u^* v J $ and $ C^*C $ are well-defined operators on $ \cD_{|C|} $, where $ J $ denotes complex conjugation. The spectrum of $ C $ is contained in the real axis for bosons and the imaginary axis for fermions.\\
Further, if $ v^*v $ has countable spectrum, then also $ C, C^*C $ and $ |C| $ have countable spectrum, see Figure \ref{fig:X_spectrum_countable}.\\
\label{lem:vvtoC}
\end{lemma}
\begin{proof}
By the Bogoliubov relations, $ C^*C = v^*v \pm (v^*v)^2 $ holds, wherever it is defined. Now, $ \lambda \mapsto \lambda \pm \lambda^2 $ is smooth apart from the critical points 0 (bosons) or 0 and 1 (fermions). The condition $ \bphi \in \cD_{|C|} $ means that the corresponding spectral function $ \phi_X $ has compact support and is smooth apart from the critical points. This property is preserved by an application of $ C^*C $, so $ C^*C: \cD_{|C|} \to \cD_{|C|} $ is well-defined. Hence, also $ |C| = \sqrt{C^*C} $ is well-defined, and by a polar decomposition also $ C = U_C |C| $ with $ U_C: \ell^2 \to \ell^2 $ unitary.\\
In the bosonic case, $ C $ is symmetric by the Bogoliubov relations, so $ C^*C = C^2 $ and further, $ \sigma(C) $ is a subset of the preimage of $ \sigma(C^2) \subseteq [0,\infty) $ under the complex map $ z \mapsto z^2 $. This preimage is contained within the real axis.\\
In the fermionic case, the Bogoliubov relations imply $ C^* = -C $, so $ C^*C = -C^2 $. That means, $ \sigma(C) $ lies within the preimage of $ \sigma(C^2) \subseteq [0,\infty) $ under the map $ z \mapsto -z^2 $, which is contained within the imaginary axis.\\
Now, suppose $ \sigma(v^*v) $ is countable. Then, $ \sigma(|C|) $ is the image of $ \sigma(v^*v) $ under the map $ z \mapsto \sqrt{z(1 \pm z)} $, which sends at most 2 arguments to the same value, so $ \sigma(|C|) $ is also countable.\\
\end{proof}

\begin{remark}
For fermions, $ u^*u + v^*v = 1 $ implies that $ v^*v $ is bounded, so $ v^*v $ and also $ u^*u $ can be defined on all of $ \ell^2 $. Hence, also $ v $ and $ u $ are defined on all of $ \ell^2 $, so the Bogoliubov relations \eqref{eq:Bogoliubovrelations} hold as a strong operator identity on $ \ell^2 $.
\end{remark}

\begin{remark}
For bosons, it is not obvious that the Bogoliubov relations \eqref{eq:Bogoliubovrelations} hold as a strong operator identity on a dense domain of $ \ell^2 $. In fact, it does not always hold as a strong operator identity on $ \cD_{|C|} $: As a counter-example, let $ v e_j := j e_j $, with respect to the canonical basis $ (e_j)_{j \in \NNN} $ and let $ u = U_u |u| $ such that $ U_u^* e_1 = c \; \sum_j j^{-1/2 - \varepsilon} e_j $ for some $ \varepsilon > 0 $. A direct calculation then shows
\begin{equation}
	u^*v e_1 = |u| U_u^* e_1 = \sum_j \sqrt{1 + j^2} j^{-1/2 - \varepsilon} e_j \;.
\end{equation}
For $ \varepsilon \le 1 $, this is obviously not in $ \ell^2 $, so $ u^*v $ is ill-defined on $ e_1 \in \cD_{|C|} $.
\end{remark}

\subsection{Extension of the Operator Algebra}
\label{subsec:extopalg}

In Section \ref{subsec:fockspace}, we defined a $ ^* $-algebra (bosonic) or $ C^* $-algebra (fermionic) $ \cA $ generated by $ a^\dagger_\pm(f), a_\pm(f), f \in \fh $. On our ITP spaces, we will encounter formal expressions in creation and annihilation operators that belong to a larger algebra $ \cAB_{\be} $, which is defined with respect to a basis $ \be = (\be_j)_{j \in \NNN} \subset \ell^2 $. We introduce the shorthand notations $ a_j := a(\be_j), a^\dagger_j := a^\dagger(\be_j) $, and consider the set of finite operator products
\begin{equation}
	\Pi_{\be} := \{a_{j_1}^\sharp \ldots a_{j_m}^\sharp \; \mid \; j_{\ell} \in \NNN \} \;.
\end{equation}
Then $ \cAB_{\be} $ is defined as the set of all complex-valued maps
\begin{equation}
	\cAB_{\be} := \{ H: \Pi_{\be} \to \CCC \} \;.
\label{eq:cAB}
\end{equation}
We formally write its elements as infinite sums
\begin{equation}
	H = \sum_{m \in \NNN} \sum_{j_1,\ldots, j_m \in \NNN} H_{j_1,\ldots, j_m} a^\sharp_{j_1} \ldots a^\sharp_{j_m} \;.
\label{eq:Hinfinitesum}
\end{equation}
$ \cAB_{\be} $ is made a $ ^* $-algebra by the involution
\begin{equation}
	^*: c a_j \mapsto \overline{c} a^\dagger_j, \quad c a^\dagger_j \mapsto \overline{c} a_j \qquad \forall c \in \CCC \;.
\end{equation}
It is easy to see that $ \cAB_{\be} $ extends $ \cA $. Resolving each $ a^\sharp(f_j) $ with respect to the basis $ \be $, we obtain a countable sum of the form \eqref{eq:Hinfinitesum}, that contains each term $ a^\sharp_{j_1} \ldots a^\sharp_{j_m} $ at most once.\\
Of particular interest will be elements of $ \cAB_{\be} $ corresponding to finite sums. For $ a^\dagger(\bphi) = \sum_{j \in \NNN} \phi_j a^\dagger_j, \bphi \in \ell^2 $, this sum is finite if and only if
\begin{equation}
	\bphi \in \cD_{\be}
	:= \{ \bphi \in \ell^2 \; \mid \; \phi_j = 0 \; \text{for all but finitely many} \; j \in \NNN\} \;.
\label{eq:cD}
\end{equation}
Here, the index $ \be $ in $ \cD_{\be} $ emphasizes that we are working with respect to the basis $ (\be_j)_{j \in \NNN} $. If $ (\be_j)_{j \in \NNN} $ is an orthonormal eigenbasis of $ |C| $, then $ \cD_{\be} = \cD_{|C|} $, since both domains are spanned by finite linear combinations of eigenvectors of $ C^*C $.\\

\subsection{Final ITP Construction and Operator Lift}
\label{subsec:finalITP}

	Next, we fix the final definitions for our ITP spaces $ \widehat{\sH} $ and define products of $ a^\dagger(\bphi), a(\bphi) $ with $ \bphi \in \cD_{\be} $ on suitable subspaces of them. Within these definitions, {\bf we assume that $ v^*v $ has countable spectrum}, so Lemma \ref{lem:vvtoC} applies and $ |C| $ has countable spectrum.\\
	As argued around \eqref{eq:muCC}, $ X \subseteq \sigma \times Y $ is countable and there exists an orthonormal eigenbasis $ (\bof_j)_{j \in \NNN} $.\\
	For bosons we follow the construction of \cite{Solovej2014}, replacing the argument ``$ C $ is Hilbert--Schmidt'' by ``$ |C| $ has countable spectrum'', which renders an orthonormal basis $ \bg = (\bg_j)_{j \in \NNN} $ as in Section~\ref{sec:bogoliubovtrafo}.	The construction uses that by Lemma \ref{lem:bogoliubovrelationssurvive}, the Bogoliubov relations still hold as a weak operator identity. Now, $ \bg $ takes the role of $ \be $ in $ \cD_{\be} $ and $ \cAB_{\be} $.\\
\begin{definition}
	The \textbf{bosonic infinite tensor product space} is given by
\begin{equation}
	\widehat{\sH} = \prod_{k \in \NNN}^\otimes \sH_k = \prod_{k \in \NNN}^\otimes \sF(\{\bg_k\}) \;,
\label{eq:bosonicITP}
\end{equation}
\label{def:bosonicITP}
\end{definition}
Note that the sequence $ (e_{k,n})_{n \in \NNN_0} $ of $ n $-particle basis vectors is a canonical basis of each one-mode Fock space $ \sH_k $, and can be used to describe elements of $ \widehat{\sH} $.\\
	For fermions the construction of \cite{Solovej2014}, with ''$ C^*C $ is trace class'' replaced by ``$ |C| $ has countable spectrum'' (which is true by Lemma \ref{lem:vvtoC}), yields an orthonormal basis $ (\boeta_j)_{j \in J} $ with countable $ J \subseteq \NNN $ as in Section \ref{subsubsec:fermioniccaseimplementation}. Here, $ \boeta $ takes the role of $ \be $ in $ \cD_{\be} $ and $ \cAB_{\be} $.\\	
The ITP construction is then a bit more delicate, since for each Cooper pair $ i \in I' $ (so $ j \in J' $), we must introduce a separate Fock space $ \sF(\{\boeta_{2i-1}\}) \otimes \sF(\{\boeta_{2i}\}) \cong \CCC^4 $ (see Remark \ref{rem:fermionicITP}). We index all $ j \in J'' $ and $ i \in I' $ by a corresponding $ k(i) $ or $ k(j) $, such that all $ k \in \NNN $ are used and take the tensor product over those $ k $:
\begin{definition}
	The \textbf{fermionic infinite tensor product space} is given by
\begin{equation}
	\widehat{\sH} = \prod_{k \in \NNN}^\otimes \sH_k = \left( \prod_{j \in J''}^\otimes \sF(\{\boeta_j\}) \right) \otimes \left( \prod_{i \in I'}^\otimes \sF(\{\boeta_{2i-1}\}) \otimes \sF(\{\boeta_{2i}\}) \right) \;.
\label{eq:fermionicITP}
\end{equation}
\label{def:fermionicITP}
\end{definition}

	For a one-mode Fock space $ \sH_{k(j)} := \sF(\{\boeta_j\}) $, the pair $ (e_{k,0}, e_{k,1}) $ forms a basis of each $ \sH_k $, while for a two-mode Fock space $ \sH_{k(i)} := \sF(\{\boeta_{2i-1}\}) \otimes \sF(\{\boeta_{2i}\}) $, such a basis is given by $ (e_{k,0,0}, e_{k,1,0}, e_{k,0,1}, e_{k,1,1}) $, with the two indices denoting the particle numbers per mode.\\

Our next challenge is to lift the one-mode creation and annihilation operators $ a_j^\dagger, a_j $ defined on the one- or two-mode Fock space $ \sH_k $ to $ \widehat{\sH} $.
\begin{lemma}
	Consider a (possibly unbounded) operator $ A_{j,j}: \sH_j \supset \dom(A_{j,j}) \to \sH_j $. Then, for $ \Psi^{(m)}_j \in \dom(A_{j,j}) $,
\begin{equation}
	A_j \Psi^{(m)} := \Psi^{(m)}_1 \otimes \ldots \otimes \Psi^{(m)}_{j-1} \otimes A_{j,j} \Psi^{(m)}_j \otimes \Psi^{(m)}_{j+1} \otimes \ldots \;,
\label{eq:operatorlift}
\end{equation}
is independent of the choice of a $ C $-sequence $ (\Psi^{(m)}) = (\Psi^{(m)}_k)_{k \in \NNN} $ representing $ \Psi^{(m)} $, and defines an operator $ A_j $ by linearity on
\begin{equation}
	\Psi \in \dom(A_j)
	:= \Big\{ \Psi = \sum_{m \in \cM} d_m \Psi^{(m)} \in \widehat{\sH} \; \Big\vert \; \Big\Vert \sum_{m \in \cM} d_m A_j \Psi^{(m)} \Big\Vert < \infty \Big\} \;,
\label{eq:domAj}
\end{equation}
where $ \cM \subseteq \NNN $, $ d_m \in \CCC $ and $ \Psi^{(m)} $ being such that $ A_j \Psi^{(m)} $ is well-defined by \eqref{eq:operatorlift}.\\
\label{lem:Operatorlift}
\end{lemma}
\begin{proof}
For a fixed choice of $ (\Psi^{(m)}) $ representing $ \Psi^{(m)} $ such that $ \Psi^{(m)}_j \in \dom(A_{j,j}) $, well-definedness of $ A_j \Psi^{(m)} $ is easy to see. By Lemma \ref{lem:psilinearcombination}, we can now represent $ \Psi = \sum_{m \in \cM} d_m \Psi^{(m)} $. And if $ \sum_{m \in \cM} d_m A_j \Psi^{(m)} $ converges, then it is independent of the representation, since $ A_j $ is linear. So $ \dom(A_j) $ and $ A_j \Psi $ are well-defined.\\
It remains to prove that $ A_k \Psi^{(m)} $ (and hence $ A_k \Psi $) is independent of the choice of $ (\Psi^{(m)}) $ representing $ \Psi^{(m)} $. For $ m \in \cM $, consider a second representative $ C $-sequence $ (\tilde\Psi^{(m)}) $ with $ \tilde\Psi^{(m)} = \Psi^{(m)} $. By Proposition \ref{prop:cksequence}, $ \tilde\Psi_k^{(m)} = c_k \Psi_k^{(m)} $ for some $ c_k \in \CCC $ with $ \prod_{k \in \NNN} c_k = 1 $. By linearity, $ A_{j,j} \tilde\Psi_k^{(m)} = c_k A_{j,j} \Psi_k^{(m)} $, so also $ A_j \Psi^{(m)} $ and $ A_j \tilde\Psi^{(m)} $ defined by \eqref{eq:operatorlift} just differ by the sequence of complex factors $ (c_k)_{k \in \NNN} $ with $ \prod_{k \in \NNN} c_k = 1 $. Hence, according to Proposition \ref{prop:cksequence}, they correspond to the same functional $ A_j \tilde\Psi^{(m)} = A_j \Psi^{(m)} $.\\
\end{proof}

The operators $ a_j, a^\dagger_j $ may be unbounded. We thus need to carefully choose a domain to make them bounded, by restricting the space of allowed $ \Psi $.
\begin{definition}
In the bosonic case, the \textbf{space $ \cS^\otimes $ with rapid decay in the particle number} is defined as
\begin{equation}
	\cS^\otimes := \Big\{ \Psi \in \bigcap_{\substack{n \in \NNN\\ k \in \NNN_0}} \dom(N_k^n) \subseteq \widehat{\sH} \; \Big\vert \; \Vert N_k^n \Psi \Vert \le c_{k,n} \Vert \Psi \Vert \; \forall k \in \NNN, \; n \in \NNN_0 \Big\} \;,
\label{eq:cSotimes}
\end{equation}
where $ c_{k,n} > 0 $ are suitable constants for each $ n $ and $ k $, and $ N_k $ is the number operator on $ \sH_k $, lifted to $ \widehat{\sH} $. The lift is possible by Lemma \ref{lem:Operatorlift}, which also allows defining $ N_k^n $.\\
In the fermionic case, the maximum particle number per mode is 1, so we always have rapid decay and simply set
\begin{equation}
	\cS^\otimes := \widehat{\sH} \;.
\end{equation}
\label{def:cSotimes}
\end{definition}

\begin{lemma}[Products of $ a^\dagger, a $ are well-defined on the ITP space]\ \\
	Consider the ITP space $ \widehat{\sH} \supseteq \cS^\otimes $ corresponding to the basis $ (\be_j)_{j \in \NNN} $, which is $ (\bg_j)_{j \in \NNN} $ (bosonic) or $ (\boeta_j)_{j \in \NNN} $ (fermionic). Then we can lift $ a_j $ and $ a_j^\dagger $ to $ \widehat{\sH} $ and for $ \bphi \in \cD_{\bg} $ or $ \cD_{\boeta} $ (defined in \eqref{eq:cD}), the expressions
	\begin{equation}
	a^\dagger(\bphi) = \sum_j \phi_j a^\dagger_j \;, \qquad
	a(\bphi) = \sum_j \overline{\phi_j} a_j \;,
\label{eq:adaggerf}
\end{equation}
define linear operators $ a^\dagger(\bphi) : \cS^\otimes \to \cS^\otimes $ and $ a(\bphi) : \cS^\otimes \to \cS^\otimes $.
\label{lem:aadaggerexist}
\end{lemma}
\begin{proof}
	First, note that the sum over $ j $ in \eqref{eq:adaggerf} is finite by definition of $ \cD_{\bg} $ and $ \cD_{\boeta} $.\\
In the fermionic case, $ a_j, a_j^\dagger $ are bounded. So by \cite[Lemma 5.1.1]{vonNeumann1939}, we can lift them to bounded operators on $ \cS^\otimes = \widehat{\sH} $, and also the finite linear combination in \eqref{eq:adaggerf} is a bounded operator on $ \cS^\otimes $.\\
In the bosonic case, where $ j = k $, we have
\begin{equation}
	\Vert a^\dagger_k \Psi \Vert = \Vert \sqrt{N_k+1} \Psi \Vert \le \Vert (N_k+1) \Psi \Vert \le (c_{k,1} + 1) \Vert \Psi \Vert \;,
\end{equation}
so
\begin{equation}
	\Vert a^\dagger(\bphi) \Psi \Vert \le \left( \sum_{k: \phi_k \neq 0} |\phi_k| (c_{k,1} + 1) \right) \Vert \Psi \Vert =: c_1 \Vert \Psi \Vert \;,
\label{eq:l1estimate}
\end{equation}
where the sum over $ k $ contains only finitely many nonzero terms, so we may call it $ c_1 > 0 $. Hence, $ a^\dagger(\bphi) \Psi \in \widehat{\sH} $ is well-defined.
It remains to establish rapid decay. Now,
\begin{equation}
\begin{aligned}
	\Vert N_k^n a^\dagger_k \Psi \Vert &= \Vert N_k^n \sqrt{N_k+1} \Psi \Vert \le \Vert N_k^{n+1} \Psi \Vert + \Vert N_k^n \Psi \Vert \le (c_{k,n+1} + c_{k,n}) \Vert \Psi \Vert\\
	&\le (c_{k,n+1} + c_{k,n}) \Vert a^\dagger_k \Psi \Vert \;,
\end{aligned}
\label{eq:Nknadagger}
\end{equation}
so by summing over $ k $, the rapid decay condition is again satisfied and $ a^\dagger(\bphi) \Psi \in \cS^\otimes $.\\

For $ a(\bphi) \Psi \in \cS^\otimes $, the same finite-sum argument can be used to obtain $ a(\bphi) \Psi \in \widehat{\sH} $. However, verifying rapid decay needs a bit more attention, since $ \Vert a^\dagger_k \Psi \Vert \le \Vert \Psi \Vert $ in \eqref{eq:Nknadagger} does not generalize to $ a_k $, see also Remark \ref{rem:ak}. However, for all $ k $ with $ a_k \Psi \neq 0 $ and $ \phi_k \neq 0 $, there is a fixed ratio $ \frac{\Vert \Psi \Vert}{\Vert a_k \Psi \Vert} =: d_k > 0 $. So with $ d := \max_k d_k $,
\begin{equation}
\begin{aligned}
	\Vert N_k^n a_k \Psi \Vert &\le (c_{k,n+1} + c_{k,n}) \Vert \Psi \Vert \le d \cdot (c_{k,n+1} + c_{k,n}) \Vert a_k \Psi \Vert \;.
\end{aligned}
\label{eq:Nkna}
\end{equation}
For $ a_k \Psi = 0 $, the inequality is trivial. A finite sum over $ k $ then establishes $ a(\bphi) \Psi \in \cS^\otimes $.\\
\end{proof}

\begin{remark} \label{rem:ak}
The condition $ \bphi \in \cD_{\bof} $ is indeed necessary, meaning we may not just allow any $ \bphi \in \ell^2 $ inside $ a^\sharp(\bphi) $, as the following counterexample shows: For the bosonic case ($ j = k $), consider $ \phi_k = \frac{1}{k} $, so $ \bphi \in \ell^2 \setminus \cD_{\bof} $. For each mode $ k $, consider the coherent state $ \Psi_k $ defined sector-wise by
\begin{equation}
	\Psi_k^{(N_k)} = e^{-\frac{\alpha_k}{2}} (N_k!)^{-\frac 12} \alpha_k^{\frac{N_k}{2}} \;,
\label{eq:coherentstate}
\end{equation}
where all $ \alpha_k \in \RRR $ are set equal to the same $ \alpha_k = \alpha > 0 $ and where $ \Vert \Psi_k \Vert_k = 1 $. Then, define the ITP $ \Psi = \prod^\otimes_{k \in \NNN} \Psi_k $. It is easy to see that $ \Psi $ satisfies the rapid decay condition \eqref{eq:cSotimes}, as for each $ \Psi_k $, $ \Vert \Psi_k^{(N_k)} \Vert_k $ decays exponentially in $ N_k $. But still, $ (\alpha_k)_{k \in \NNN} \notin \ell^2 $, so we may think of $ \Psi $ as a ``coherent state with a large displacement'', living outside the Fock space. It is a well-known fact about coherent states that $ a_k \Psi = \alpha \Psi $, so
\begin{equation}
	\Vert a(\bphi) \Psi \Vert
	= \Big\Vert \sum_k \overline{\phi_k} \alpha \Psi \Big\Vert
	= \alpha \sum_k \frac{1}{k} \Vert \Psi \Vert
	= \infty \;.
\end{equation}
Hence, $ a(\bphi) $ is ill-defined on $ \Psi $.\\
	The same happens with any coherent state product \eqref{eq:coherentstate} and any $ \bphi $, where $ \sum_k \overline{\phi_k} \alpha_k = \infty $. In particular, the space of allowed $ (\phi_k)_{k \in \NNN} $ is dual to the one of allowed $ (\alpha_k)_{k \in \NNN} $.\\
\end{remark}

\begin{remark}
It is also possible to define $ a^\sharp(\bphi) $ for more general $ \bphi $, if one restricts the space $ \cS^\otimes $. For instance, it is not too difficult to see that imposing uniform rapid decay via
\begin{equation}
	\Psi \in \cS^\otimes_{\uni}
	:= \left\{ \Psi \in \widehat{\sH} \; \middle\vert \; \Vert N_k^n \Psi \Vert \le c_n \Vert \Psi \Vert \; \forall n, k \in \NNN \right\}
\label{eq:cSotimesuni}
\end{equation}
yields a well-defined product $ a^\sharp(\bphi_1) \ldots a^\sharp(\bphi_n) \Psi \in \widehat{\sH} $ for any $ \bphi_1, \ldots, \bphi_n \in \ell^1 $.\\
Alternatively, one may allow for $ \bphi_1, \ldots, \bphi_n \in \ell^p $, $ p \in (1,2] $ by restricting to
\begin{equation}
	\Psi \in \cS^\otimes_q
	:= \Big\{ \Psi \in \widehat{\sH} \; \Big\vert \; \Vert (N_k+1)^{n/2} \Psi \Vert \le c_k^n \Vert \Psi \Vert \; \text{with} \; \sum_k c_k^q < \infty \quad \forall n \in \NNN \Big\} \;,
\label{eq:cSotimesq}
\end{equation}
where $ q $ is the Hölder dual of $ p $, i.e., $ \frac 1p + \frac 1q = 1 $. Note that $ \cS^\otimes_{\uni} \subset \cS^\otimes_q \subset \cS^\otimes $.\\
\end{remark}

\begin{remark}
It is possible to view the subspace $ \prod^{\otimes C}_{k \in \NNN} \sH_k $ of the equivalence class $ C $ as the \textit{original Fock space} with respect to the vacuum $ \Omega = \prod^\otimes_{k \in \NNN} e_{k,0} $: Recall that each $ \Psi \in \prod^{\otimes C}_{k \in \NNN} \sH_k $ can be written in coordinates as \eqref{eq:Phiincoordinates}:
\begin{equation}
	\Psi = \sum_{n(\cdot) \in F} a(n(\cdot)) \prod^\otimes_{k \in \NNN} e_{k,n(k)} \;,
\label{eq:Phiinspecificcoordinates}
\end{equation}
with $ F $ containing all sequences $ (n(k))_{k \in \NNN} $, such that $ n(k)=0 $ for almost all $ k $. Hence, each $ \prod^\otimes_{k \in \NNN} e_{k,n(k)} $ is a tensor product state of finitely many particles. Since the Fock norm and the $ \widehat{\sH} $-norm coincide, the vector $ \prod^\otimes_{k \in \NNN} e_{k,n(k)} $ can be seen as a Fock space vector normalized to 1. The linear combination \eqref{eq:Phiinspecificcoordinates} with $ \sum_{n(\cdot)} |a(n(\cdot))|^2 $ can hence also be interpreted a Fock space vector.\\
Conversely, each Fock space vector can be written as a countable sequence \eqref{eq:Phiinspecificcoordinates}, since the span of the above-mentioned tensor product states is dense in $ \sF $.\\
\end{remark}

\section{Implementation: Extended}
\label{sec:implementation}

We proceed with defining implementability of $ \cV $ by an extended operator $ \UUU_\cV $ on $ \widehat{\sH} $. Lemma \ref{lem:UUUcV} establishes that $ \UUU_\cV $ is well-defined and Lemma \ref{lem:implementation} gives conditions for when $ \UUU_\cV $ is an implementer in the extended sense.
We then prove our main results by verifying these conditions in Theorem \ref{thm:bosoniccountable} for bosons and Theorem \ref{thm:fermioniccountable} for fermions.

\subsection{Definition of Extended Implementation}
\label{subsec:implementationdef}

The implementer $ \UUU_\cV $ is defined on a dense subspace of Fock space $ \cD_\sF \subset \sF $, that contains a finite number of particles from the space $ \cD_{\bof} $ (defined by \eqref{eq:cD} with $ \be = \bof $):
\begin{equation}
	\cD_\sF := \mathrm{span}\{a^\dagger(\bphi_1) \ldots a^\dagger(\bphi_N) \Omega, \; N \in \NNN_0, \; \bphi_{\ell} \in \cD_{\bof}\} \;.
\label{eq:sFfin}
\end{equation}
	The operator $ \UUU_\cV $ now maps from $ \cD_\sF $ into an ITP space $ \widehat{\sH} $. Note that due to Lemma \ref{lem:aadaggerexist}, the expressions $ b^\dagger(\bphi) := a^\dagger(u \bphi) + a(v \overline{\bphi}) $ and $ b(\bphi) := a(u \bphi) + a^\dagger(v \overline{\bphi}) $ define operators $ \cS^\otimes \to \cS^\otimes $, i.e., in an extended sense.\\

\begin{definition}
	We say that a linear operator $ \UUU_\cV: \cD_\sF \to \widehat{\sH} $ \textbf{implements} a Bogoliubov transformation $ \cV $ \textbf{in the extended sense}, if for all $ \bphi \in \cD_{\bof}, \; \Psi \in \UUU_\cV[\cD_\sF] $, we have that
\begin{equation}
	\UUU_\cV a^\dagger(\bphi) \UUU_\cV^{-1} \Psi = b^\dagger(\bphi) \Psi \;, \qquad
	\UUU_\cV a(\bphi) \UUU_\cV^{-1} \Psi = b(\bphi) \Psi \;.
\label{eq:implementation}
\end{equation}
\label{def:implementation}
\end{definition}
This requires, of course, that $ \UUU_\cV^{-1} $ is well-defined. So before establishing \eqref{eq:implementation}, we have to show that $ \UUU_\cV $ is invertible. This will be one main difficulty within the upcoming proofs.\\
The implementer $ \UUU_\cV $ is constructed as follows: First we define some new vacuum vector $ \Omega_\cV = \UUU_\cV \Omega \in \cS^\otimes $, such that
\begin{equation}
	b(\bphi) \Omega_\cV = 0 \;.
\label{eq:vacuumannihilation}
\end{equation}

Then we make $ \UUU_\cV $ change $ a^\sharp $- into $ b^\sharp $-operators:
\begin{definition}
Given a Bogoliubov transformed vacuum state $ \Omega_\cV \in \cS^\otimes $, the \textbf{Bogoliubov implementer} $ \UUU_\cV $ is formally defined on $ \cD_\sF $ by
\begin{equation}
	\UUU_\cV a^\dagger(\bphi_1) \ldots a^\dagger(\bphi_n) \Omega := b^\dagger(\bphi_1) \ldots b^\dagger(\bphi_n) \Omega_\cV \;,
\label{eq:transformbogoliubovstate}
\end{equation}
with $ \bphi_{\ell} \in \cD_{\bof} $ and $ b^\dagger(\bof_j) = (a^\dagger(u \bof_j) + a(v \overline{\bof_j})) $ for all basis vectors $ \bof_j $ in $ \bof $.\\
\label{def:transformbogoliubovstate}
\end{definition}

\begin{lemma}[$ \UUU_\cV $ is well-defined]\ \\
	If $ \Omega_\cV \in \cS^\otimes \subseteq \widehat{\sH} $ (see \eqref{eq:cSotimes}), then \eqref{eq:transformbogoliubovstate} defines an operator $ \UUU_\cV: \cD_\sF \to \cS^\otimes $.\\
\label{lem:UUUcV}
\end{lemma}

\begin{proof}
	Both $ u \bof_j $ and $ v \overline{\bof_j} $ are proportional to the same basis vector $ \be_j $ (bosonic: $ \bg_j $, fermionic: $ \boeta_j $, see~\cite{Solovej2014}). So the right-hand side of \eqref{eq:transformbogoliubovstate} is a finite linear combination of vectors $ a^\sharp(\be_{j_1}) \ldots a^\sharp(\be_{j_n}) \Omega_\cV $. Now, $ \Omega_\cV \in \cS^\otimes $ and by Lemma \ref{lem:aadaggerexist}, each application of $ a^\sharp(\be_j) $ leaves the vector in $ \cS^\otimes $.\\
\end{proof}

\begin{lemma}[Conditions for an implementer $ \UUU_\cV $]
	Suppose that for a Bogoliubov transformation (i.e., $ \cV $ satisfying \eqref{eq:Bogoliubovrelations}) an $ \Omega_\cV $ satisfying $ b(\bphi) \Omega_\cV = 0 $ for all $ \bphi \in \cD_{\bof} \subseteq \ell^2 $ has been found, such that $ \UUU_\cV $ in \eqref{eq:transformbogoliubovstate} is well-defined on $ \cD_\sF $ and has an inverse $ \UUU_\cV^{-1} $ defined on $ \UUU_\cV[\cD_\sF] $. Then, $ \UUU_\cV $ implements $ \cV $ in the sense of \eqref{eq:implementation} on all $ \Psi \in \UUU_\cV[\cD_\sF] $.\\
\label{lem:implementation}
\end{lemma}
\begin{proof}
	We write $ \Psi = \UUU_\cV \Phi $ with $ \Phi \in \cD_\sF $. By linearity, it suffices to prove the statement for $ \Phi = a^\dagger(\bphi_1) \ldots a^\dagger(\bphi_n) \Omega $, which implies by \eqref{eq:transformbogoliubovstate} that $ \Psi = b^\dagger(\bphi_1) \ldots b^\dagger(\bphi_n) \Omega_\cV $. Checking the first statement of \eqref{eq:implementation}, i.e., $ \UUU_\cV a^\dagger(\bphi) \UUU_\cV ^{-1} \Psi= b^\dagger(\bphi) \Psi $ is straightforward. For the second statement, we make use of the CAR/CCR of $ a $- and $ b $-operators, using $ \varepsilon = (-1) $ for fermions and $ \varepsilon = 1 $ for bosons. Here, the CAR/CCR are valid for $ a $-operators by definition, and for $ b $-operators, since by means of Lemma \ref{lem:bogoliubovrelationssurvive}, the Bogoliubov relations survive the extension.\\
\begin{equation}
\begin{aligned}
	\UUU_\cV a(\bphi) \UUU_\cV ^{-1} \Psi
	&= \sum_{\ell = 1}^n \UUU_\cV a^\dagger(\bphi_1) \ldots a^\dagger(\bphi_{\ell-1}) \varepsilon^{\ell+1} \langle \bphi, \bphi_\ell \rangle a^\dagger(\bphi_{\ell+1}) \ldots a^\dagger(\bphi_n) \Omega\\
	&= b(\bphi) b^\dagger(\bphi_1) \ldots b^\dagger(\bphi_n) \Omega_\cV - \varepsilon^{n+1} b^\dagger(\bphi_1) \ldots b^\dagger(\bphi_n) b(\bphi) \Omega_\cV\\
	&\overset{\eqref{eq:vacuumannihilation}}{=} b(\bphi) b^\dagger(\bphi_1) \ldots b^\dagger(\bphi_n) \Omega_\cV = b(\bphi) \Psi \;,\\
\end{aligned}
\end{equation}
where we used the convention that the above sums are set to zero for $ N = 0 $.\\

\end{proof}

\subsection{Bosonic Case}
\label{subsec:countablevboson}

We now show that for a suitable choice of $ \Omega_\cV $, the operator $ \UUU_\cV $ defined in \eqref{eq:transformbogoliubovstate} indeed implements the Bogoliubov transformation $ \cV $.\\

\begin{theorem}[Implementation via ITP works, bosonic]
	 Consider a bosonic Bogoliubov transformation $ \cV = \left( \begin{smallmatrix} u & v \\ \overline{v} & \overline{u} \end{smallmatrix} \right) $ with $ v^*v $ having countable spectrum. Let $ \widehat{\sH} = \prod_{k \in \NNN}^\otimes \sH_k $ be the ITP space (Definition \ref{def:bosonicITP}) with respect to the basis $ (\bg_k)_{k \in \NNN} \subset \ell^2 $. Define the new vacuum vector
\begin{equation}
	\Omega_\cV = \prod_{k \in \NNN}^\otimes \Omega_{k,\cV} := \prod_{k \in \NNN}^\otimes \left( \left( \left(1 - \tfrac{\nu_k^2}{\mu_k^2} \right)^{1/4} \right) \exp{\left( -\tfrac{\nu_k}{2 \mu_k} (a^\dagger(\bg_k))^2 \right) } \Omega_k \right) \;,
\label{eq:bosonicbogoliubovvacuumcountable}
\end{equation}
	where $ \mu_k, \nu_k $ are the singular values of $ u, v $ as in Section \ref{subsec:implementation1}. Then, $ \cV $ is implemented in the sense of \eqref{eq:implementation} by $ \UUU_\cV: \cD_\sF \to \widehat{\sH} $ \eqref{eq:transformbogoliubovstate}.\\
\label{thm:bosoniccountable}
\end{theorem}

\begin{proof}
By Lemma \ref{lem:implementation}, we need to establish the following four points:
\begin{enumerate}
\item The new vacuum $ \Omega_\cV $ is well-defined
\item $ \UUU_\cV $ is well-defined on $ \cD_\sF $ (Lemma \ref{lem:UUUcV} will be used, here)
\item $ b(\bphi) \Omega_\cV = 0 $
\item $ \UUU_\cV^{-1} $ exists on $ \UUU_\cV [\cD_\sF] $
\end{enumerate}

\noindent \underline{1.) Well-definedness of $ \Omega_\cV $}: Expression \eqref{eq:bosonicbogoliubovvacuumcountable} is an ITP of one normalized factor per space $ \sH_k $. Hence, it is a $ C $-sequence, which can be identified with $ \Omega_\cV \in \widehat{\sH} $.\\

\noindent \underline{2.) Well-definedness of $ \UUU_\cV $}: follows from Lemma \ref{lem:UUUcV}, if we can establish $ \Omega_\cV \in \cS^\otimes $. By definition of $ \cS^\otimes $, we need to verify the rapid decay condition $ \Vert N_k^n \Omega_\cV \Vert \le c_{k,n} \Vert \Omega_\cV \Vert $. Since all $ \Omega_{k,\cV} $ are normalized, this boils down to proving $ \Vert N_k^n \Omega_{k,\cV} \Vert_k^2 \le c_{k,n}^2 $. We explicitly compute
\begin{equation}
	\Vert N_k^n \Omega_{k,\cV} \Vert_k^2 = (1-4t^2)^{1/2} \sum_{N = 0}^\infty \frac{t^{2N} (2N)!}{(N!)^2} (2N)^{2n} \;,
\label{eq:rapiddecaycheck}
\end{equation}
with $ t = \left\vert \tfrac{\nu_k}{2 \mu_k} \right\vert \in [0,1/2) $. Now, the function
\begin{equation}
	N \mapsto \frac{t^{2N} (2N)!}{(N!)^2} (2N)^{2n} \le (2t)^{2N} (2N)^{2n} \;,
\end{equation}
is positive, bounded and decays exponentially at $ N \to \infty $ since $ 0 \le 2t < 1 $. So,
\begin{equation}
	\sum_{N = 0}^\infty \frac{t^{2N} (2N)!}{(N!)^2} (2N)^{2n} \le \text{cons.} + \sum_{N = 0}^\infty (2t)^{2N} (2N)^{2n} =: c_{k,n}^2 < \infty \;,
\end{equation}
which establishes $ \Omega_\cV \in \cS^\otimes $ and hence the claim.\\

\noindent \underline{3.) $ b(\bphi) $ annihilates $ \Omega_\cV $}: This is straightforward to check: Since $ \bphi \in \cD_{\bof} $, the following sum over $ k $ is finite:
\begin{equation}
	b(\bphi) \Omega_\cV = \sum_k \overline{\phi_k} b(\bof_k) \Omega_\cV.
\end{equation}
As in the case, where the Shale condition holds, each $ b(\bof_k) $ annihilates the corresponding vacuum vector $ \Omega_{k,\cV} $, so the finite sum above is 0.\\

\noindent \underline{4.) Well-definedness of $ \UUU_\cV^{-1} $}: Consider the basis $ \{a^\dagger(\bof_{k_1}) \ldots a^\dagger(\bof_{k_N}) \Omega \; \mid \; N \in \NNN_0, \; k_{\ell} \in \NNN \} $ of $ \cD_\sF $, where $ \bof_{k_{\ell}} $ are chosen out of the basis $ (\bof_j)_{j \in \NNN} $ (with $ j = k $). If we can show that the set
\begin{equation}
	\{b^\dagger(\bof_{k_1}) \ldots b^\dagger(\bof_{k_N}) \Omega_\cV \; \mid \; N \in \NNN_0, \; k_{\ell} \in \NNN \} \subset \widehat{\sH} \;,
\label{eq:bsetbosonic}
\end{equation}
with $ b^\dagger(\bof_k) = \mu_k a^\dagger(\bg_k) + \nu_k a(\bg_k) $, is linearly independent, we are done, since then $ \mathrm{ker}(\UUU_\cV) = \{0\} $, so $ \UUU_\cV $ is injective and hence invertible on its image.\\
	Now, as applications of $ b^\dagger_k := b^\dagger(\bof_k) $ and $ \UUU_\cV $ preserve the ITP structure, it suffices to show that on each mode $ k $, the set
\begin{equation}
	\{(b^\dagger_k)^N \Omega_{k,\cV} \; \mid \; N \in \NNN_0 \} \subset \sH_k
\label{eq:bNset}
\end{equation}
is linearly independent. But \eqref{eq:bNset} is just the image of the set
\begin{equation}
	\{(a^\dagger(\bof_k))^N \Omega_k \; \mid \; N \in \NNN_0 \} \subset \sF(\{\bof_k\})
\label{eq:aNset}
\end{equation}
under a one-mode Bogoliubov transformation $ \UUU_{k,\cV}: \sF(\{\bof_k\}) \to \sH_k $ (defined as $ \UUU_{j,\cV} $ in \eqref{eq:bosonicbogoliubovimplementer}). For a finite number $ m $ of modes, Bogoliubov transformations can always be implemented by unitary operators, as then the operator $ v: \CCC^m \to \CCC^m $ is always Hilbert--Schmidt. Now, \eqref{eq:aNset} is an orthogonal set with no vector being 0, so its image \eqref{eq:bNset} under $ \UUU_{k,\cV} $ is also orthogonal with no vector being 0, and hence it is linearly independent.\\
\end{proof}

\subsection{Fermionic Case}
\label{subsec:countablevfermion}

\begin{theorem}[Implementation via ITP works, fermionic]
	Consider a fermionic Bogoliubov transformation $ \cV = \left( \begin{smallmatrix} u & v \\ \overline{v} & \overline{u} \end{smallmatrix} \right) $ with $ v^*v $ having countable spectrum. Let $ \widehat{\sH} = \prod_{k \in \NNN}^\otimes \sH_k $ be the ITP space (Definition \ref{def:fermionicITP}). Define the new vacuum vector
\begin{equation}
\begin{aligned}
	\Omega_\cV 
	= &\prod_{j \in J''}^\otimes \Omega_{j,\cV} \otimes \prod_{i \in I'}^\otimes \Omega_{2i,2i-1,\cV} \\
	:= &\left(\prod_{j \in J''_1}^\otimes a^\dagger(\boeta_j) \Omega_j \right) \otimes \left(\prod_{j \in J''_0}^\otimes \Omega_j \right) \otimes \left( \prod_{i \in I'}^\otimes (\alpha_i - \beta_i a^\dagger(\boeta_{2i}) a^\dagger(\boeta_{2i-1})) \Omega_{2i,2i-1} \right) \;,
\end{aligned}
\label{eq:fermionicbogoliubovvacuum}
\end{equation}
	with $ \alpha_i, \beta_i $ being the singular values of $ u, v $ as in \eqref{eq:fermionictrafo3}, and with $ \Omega_{2i,2i-1}, \Omega_{2i,2i-1,\cV} \in \sH_{k(i)} $. Then, $ \cV $ is implemented in the sense of \eqref{eq:implementation} by $ \UUU_\cV: \cD_\sF \to \widehat{\sH} $ \eqref{eq:transformbogoliubovstate}.\\
\label{thm:fermioniccountable}
\end{theorem}

\begin{proof}

Again, by Lemma \ref{lem:implementation}, it suffices to establish the four points in the proof of the bosonic case (Theorem \ref{thm:bosoniccountable}). Points 1.) and 3.) are analogous, while 2.) follows from Lemma \ref{lem:UUUcV} as $ \cS^\otimes = \widehat{\sH} $.\\

\noindent \underline{4.) Well-definedness of $ \UUU_\cV^{-1} $}: We proceed as in proof step 4.) in Theorem \ref{thm:bosoniccountable}. So we are done if we can prove that the set
\begin{equation}
	\{b^\dagger(\bof_{j_1}) \ldots b^\dagger(\bof_{j_N}) \Omega_\cV \; \mid \; N \in \NNN_0, \; j_{\ell} \in J \} \subset \widehat{\sH}
\label{eq:bsetfermionic}
\end{equation}
is linearly independent. This again boils down to proving a linear independence statement on each $ \sH_k $. The crucial difference now is, that each tensor product factor $ \sH_k $ may be a Fock space over either one or two modes. We abbreviate $ b^\sharp_j := b^\sharp(\bof_j) $ and $ a^\sharp_j := a^\sharp(\boeta_j) $. For two-mode factors indexed by $ i \in I' $, we need to prove linear independence of the set
\begin{equation}
	\{(b^\dagger_{2i})^{N_1} (b^\dagger_{2i-1})^{N_2} \Omega_{k(i),\cV} \; \mid \; N_1,N_2 \in \{0, 1\} \} \subset \sH_{k(i)} \;.
\label{eq:bN1N2set}
\end{equation}
This follows from $ \UUU_{2i,2i-1,\cV} $ (see \eqref{eq:fermionicbogoliubovimplementer}) being unitary and mapping the set
\begin{equation}
	\{(a^\dagger_{2i})^{N_1} (a^\dagger_{2i-1})^{N_2} \Omega_{k(i)} \; \mid \; N_1,N_2 \in \{0, 1\} \} \subset \sF(\{\bof_{2i}\}) \otimes \sF(\{\bof_{2i-1}\})
\label{eq:aN1N2set}
\end{equation}
onto \eqref{eq:bN1N2set}. For one-mode factors indexed by $ j \in J'' $ we need linear independence of
\begin{equation}
	\{(b^\dagger_j)^N \Omega_{k(j),\cV} \; \mid \; N \in \{0, 1\} \} \subset \sH_{k(j)} \;.
\label{eq:bNsetonemode}
\end{equation}
This follows again by unitarity of $ \UUU_{j,\cV} $, as well as orthogonality and zero-freeness of the set $ \{(a^\dagger_j)^N \Omega_{k(j)} \; \mid \; N \in \{0, 1\} \} \subset \sF(\{\bof_j\}) $, which is mapped to \eqref{eq:bNsetonemode}. By linear independence of \eqref{eq:bN1N2set} and \eqref{eq:bNsetonemode}, we obtain linear independence of \eqref{eq:bsetfermionic}, which implies injectivity of $ \UUU_\cV $ and finishes the proof.\\
\end{proof}

\begin{remark} \label{rem:ITPunitary}
We may extend $ \UUU_\cV $ to a unitary operator on $ \widehat{\sH} $: Both in the bosonic and the fermionic case, we have $ \widehat{\sH} = \prod_{k \in \NNN}^\otimes \sH_k $, compare \eqref{eq:bosonicITP} and \eqref{eq:fermionicITP}, where a unitary $ \UUU_{k, \cV}: \sH_k \to \sH_k $ is defined in \eqref{eq:bosonicbogoliubovimplementer} (bosonic) and \eqref{eq:fermionicbogoliubovimplementer} (fermionic). On tensor product states $ \Psi^{(m)} := \prod_k^\otimes \Psi_k^{(m)} $, we can thus immediately define $ \UUU_\cV $ as
\begin{equation}
	\UUU_\cV \Psi^{(m)} := \prod_k^\otimes (\UUU_{k, \cV} \Psi_k^{(m)}) \;, \qquad
	\Vert \UUU_\cV \Psi^{(m)} \Vert
	= \Vert \Psi^{(m)} \Vert \;.
\end{equation}
By Lemma \ref{lem:psilinearcombination}, there exists an orthonormal set $ \{\Psi^{(m)}\}_{m \in \cM} $ such that any $ \Psi \in \widehat{\sH} $ can be written as a convergent series $ \Psi = \sum_{m \in \cM} d_m \Psi^{(m)} $, $ d_m \in \CCC $. One easily checks that $ \UUU_\cV $ preserves orthonormality of $ \Psi^{(m)} $. Thus, $ \UUU_\cV $ first extends to a bounded operator on all $ \Psi $ that are finite linear combinations of $ \Psi^{(m)} $, and then to all $ \Psi \in \widehat{\sH} $ by continuity. Unitarity of $ \UUU_\cV: \widehat{\sH} \to \widehat{\sH} $ is then straightforward to check.\\
\end{remark}

\begin{remark} \label{rem:fermionicITP}
It is crucial that the fermionic ITP space has been chosen as $ \widehat{\sH} = \prod_k^\otimes \sH_k $, with two-mode spaces $ \sH_k = \sF(\{\boeta_{2i}\}) \otimes \sF(\{\boeta_{2i-1}\}) $ for Cooper pairs $ i \in I' $. If we had just chosen a product of one-mode spaces $ \prod_{j \in J}^\otimes \sF(\{\boeta_j\}) $, then invertibility of $ \UUU_\cV $ may fail.\\
As an example, consider a $ \cV $ with countably infinitely many Cooper pairs $ i \in I' $, such that $ \alpha_i = \beta_i = \frac{1}{\sqrt{2}} $. Then, each Cooper pair is in the state
\begin{equation}
	\Psi_i := \frac{1}{\sqrt{2}} (| 0 \rangle \otimes | 0 \rangle + | 1 \rangle \otimes | 1 \rangle) \in \CCC^4 \;,
\label{eq:cooperstate}
\end{equation}
i.e., we have a ``half particle--hole transformation''. When evaluating the formal ITP $ \Omega_\cV = \prod_{i \in I'}^\otimes \Psi_i $, we obtain a sum of $ C $-sequences: For each pair $ i $, one has to choose either $ | 0 \rangle \otimes | 0 \rangle $ or $ | 1 \rangle \otimes | 1 \rangle $ as a contribution to $ \Omega_\cV $ and sum over all choices. But now, there are uncountably many such choices, as each corresponds to a binary number of infinitely many digits. And each one gives a contribution of norm $ \prod_{i \in I'} \frac{1}{\sqrt{2}} = 0 $. So $ \Omega_\cV = 0 $, making $ \UUU_\cV $ non-invertible.\\
\end{remark}

\section{Diagonalization: Extended}
\label{sec:diagonalization}

As we now have conditions for $ \cV $ being implementable by $ \UUU_\cV $ in the extended sense, it would be interesting to diagonalize quadratic Hamiltonians $ H $ via $ \UUU_\cV $. We now give precise definitions of ``quadratic Hamiltonian'' and ``diagonalized'' in the extended sense and provide diagonalizability criteria in Propositions \ref{prop:diagonalizableboson} and \ref{prop:diagonalizablefermion}.\\

\subsection{Definition of Extended Diagonalization}
\label{subsec:diagonalizationextdef}

Recall that the extended operator algebra $ \cAB_{\be} $, defined in \eqref{eq:cAB}, consists of all maps $ H $ that assign to each finite operator product $ a^\sharp_{j_1} \ldots a^\sharp_{j_m} $ a complex coefficient $ H_{j_1,\ldots, j_m} \in \CCC $. Each map $ H $ can be interpreted as a (possibly infinite) sum
\begin{equation}
	H = \sum_m \sum_{j_1,\ldots, j_m} H_{j_1,\ldots, j_m} a^\sharp_{j_1} \ldots a^\sharp_{j_m} \;.
\end{equation}

A \textbf{formal quadratic Hamiltonian} is an element $ H \in \cAB_{\be} $, where $ H_{j_1,\ldots, j_m} \neq 0 $ only appears for $ m = 2 $ and $ H^* = H $. We impose normal ordering on quadratic
Hamiltonians (see Remark \ref{rem:normalordering}), so they read
\begin{equation}
	H = \frac{1}{2} \sum_{j,k \in \NNN} (2 h_{jk} a^\dagger_j a_k \pm k_{jk} a^\dagger_j a^\dagger_k + \overline{k_{jk}} a_j a_k) \;,
\label{eq:Hquadratic}
\end{equation}
where $ \pm $ means $ + $ in the bosonic and $ - $ in the fermionic case. The term ``formal'' stresses that $ H $ is not necessarily an operator on Fock space. To $ H $ we associate a block matrix
\begin{equation}
	A_H = \begin{pmatrix} h & \pm k \\ \overline{k} & \pm \overline{h} \end{pmatrix} \;,
\label{eq:AHquadratic}
\end{equation}
with $ h = (h_{jk})_{j,k \in \NNN} $, $ k = (k_{jk})_{j,k \in \NNN} $ being matrices of infinite size.\\
Consider a Bogoliubov transformation $ \cV = \left( \begin{smallmatrix} u & v \\ \overline{v} & \overline{u} \end{smallmatrix} \right) $. Then, $ a^\dagger_j \mapsto b^\dagger_j = \sum_k (u_{jk} a^\dagger_k + \overline{v_{jk}} a_k) $, followed by a \textbf{normal ordering}, defines a corresponding algebraic Bogoliubov transformation $ \cV_{\cAB} : \cAB_{\be} \supseteq \dom(\cV_{\cAB}) \to \cAB_{\be} $, where $ \dom(\cV_{\cAB}) \subseteq \cAB_{\be} $ is a suitable subspace that avoids diverging sums over $ k $. The transformed operator and its associated block matrices are
\begin{equation}
	\widetilde{H} = \cV_{\cAB}(H) \;, \qquad 
	A_{\widetilde{H}} = \cV^* A_H \cV \;.
\end{equation}

\begin{definition}
	A formal quadratic Hamiltonian $ H \in \cAB_{\be} $ is called \textbf{diagonalizable in the extended sense} if there exists a Bogoliubov transformation $ \cV $, such that
\begin{equation}
	\cV^* A_H \cV = \begin{pmatrix} E & 0 \\ 0 & \pm E \end{pmatrix} \;,
\end{equation}
with $ E \ge 0 $ being Hermitian, $ \pm $ being $ + $ in the bosonic and $ - $ in the fermionic case, and where $ \cV $ is implementable in the extended sense (see Definition \ref{def:implementation}).\\
\label{def:diagonalizable}
\end{definition}

The Hamiltonian associated with $ A_{\widetilde{H}} $ is then $	\widetilde{H} = \di \Gamma(E) $, where the matrix $ E $ provides a positive semidefinite quadratic form on $ \cD_{\be} $. So by Friedrichs' theorem, it has a self-adjoint extension on $ \dom(E) $. Following \cite[Sect. VIII.10]{reedsimon1}, $ \di \Gamma(E) $ is then essentially self-adjoint on $ \bigoplus_{n = 0}^\infty \dom(E)^{\otimes n} \subseteq \sF $, so $ \widetilde{H} $ defines quantum dynamics on $ \sF $.\\

\begin{remark}[Normal ordering constant] \label{rem:normalordering}
Our process of ``diagonalizing'' a Hamiltonian $ H $ actually consists of conjugating it with $ \UUU_\cV $, so $ a^\sharp $ is replaced by $ b^\sharp $, \textit{plus a subsequent normal ordering process}. This process is equivalent to adding a constant to the Hamiltonian, namely
\begin{equation}
	c = \frac{1}{2} \left( \tr(E) - \tr(h) \right) = \frac{1}{2} \sum_j (E_{jj} - h_{jj}) \;.
\end{equation}
The sum might be divergent and hence not a complex number. Nevertheless, using ESS (see Appendix \ref{app:ESS}), we may interpret it as an infinite renormalization constant $ c \in \Ren_1(\NNN) $, namely the one associated with the sequence $ c_j = \frac{1}{2} (E_{jj} - h_{jj}) $.
If $ E $ now maps $ \cD_{\be} $ into itself (so each column has finitely many non-zero entries), then
\begin{equation*}
	\widetilde{H} = \UUU_\cV^{-1} (H + c) \UUU_\cV \;,
\end{equation*}
which is in accordance with \eqref{eq:renormalizedhamiltonian}. Otherwise, we may decompose $ H = \sum_{n \in \NNN} H^{(n)} $ and $ c = \sum_{n \in \NNN} c^{(n)} $, such that in $ \cV^* A_{H^{(n)}} \cV = \left( \begin{smallmatrix} E^{(n)} & 0 \\ 0 & \pm E^{(n)} \end{smallmatrix} \right) $, each $ E^{(n)} $ maps $ \cD_{\be} $ into itself. So
\begin{equation*}
	\widetilde{H} = \sum_{n \in \NNN} \UUU_\cV^{-1} (H^{(n)} + c^{(n)}) \UUU_\cV \;,
\end{equation*}
which is in accordance with \eqref{eq:renormalizedhamiltonian2}.
\end{remark}

\subsection{Bosonic Case}
\label{subsec:diagonalizationboson}

Conditions for the existence of a $ \cV $, such that $ \cV^* A_H \cV $ is block-diagonal, can be found in \cite[Thms. 1 and 4]{NamNapiorkowskiSolovej2016}. We can use them to readily derive conditions for when a formal quadratic Hamiltonian $ H $ is diagonalizable in the extended sense:

\begin{proposition}[Extended diagonalizability, bosonic case]
	Let a formal quadratic bosonic Hamiltonian $ H $ \eqref{eq:Hquadratic} be given such that for the associated block matrix $ A_H $ \eqref{eq:AHquadratic} we have $ h > 0 $, and that $ G = h^{-1/2} k h^{-1/2} $ is a bounded operator with $ \Vert G \Vert < 1 $. Following \cite[Thm. 1]{NamNapiorkowskiSolovej2016}, there exists a bosonic Bogoliubov transformation $ \cV = \left( \begin{smallmatrix} u & v \\ \overline{v} & \overline{u} \end{smallmatrix} \right) $ such that
\begin{equation}
	\cV^* A_H \cV = \begin{pmatrix} E & 0 \\ 0 & E \end{pmatrix} \;.
\end{equation}	
Suppose further that $ v^*v $ has countable spectrum.\\
	Then, $ H $ is diagonalizable in the extended sense.\\
\label{prop:diagonalizableboson}
\end{proposition}

\begin{proof}
	Consider Definition \ref{def:diagonalizable} for diagonalizability. The existence of $ \cV $ as a block matrix associated with a bounded operator on $ \ell^2 $ is a direct consequence of \cite[Thm. 1]{NamNapiorkowskiSolovej2016}.\\
If the spectrum of $ v^*v $ is countable, then implementability of $ \cV $ in the extended sense follows from Theorem \ref{thm:bosoniccountable}.\\
\end{proof}

\subsection{Fermionic Case}
\label{subsec:diagonalizationfermion}

\begin{proposition}[Extended diagonalizability, fermionic case]
	Let a formal quadratic fermionic Hamiltonian $ H $ \eqref{eq:Hquadratic} be given such that for the associated block matrix $ A_H $ \eqref{eq:AHquadratic}, $ \mathrm{dimKer}(A_H) $ is even or $ \infty $. Following \cite[Thm. 4]{NamNapiorkowskiSolovej2016}, there exists some fermionic Bogoliubov transformation $ \cV = \left( \begin{smallmatrix} u & v \\ \overline{v} & \overline{u} \end{smallmatrix} \right) $ such that
\begin{equation}
	\cV^* A_H \cV = \begin{pmatrix} E & 0 \\ 0 & - E \end{pmatrix} \;.
\end{equation}
Suppose further that $ v^*v $ has countable spectrum.\\
	Then, $ H $ is diagonalizable in the extended sense on the ITP space $ \widehat{\sH} $.\\
\label{prop:diagonalizablefermion}
\end{proposition}

\begin{proof}
	Existence of a unitary $ \cV $ and of $ E \ge 0 $ follows from \cite[Thm. 4]{NamNapiorkowskiSolovej2016}. By unitarity, $ \cV^* \cV = 1 = \cV \cV^* $, so $ \cV $ is a fermionic Bogoliubov transformation.
If $ \sigma(v^*v) $ is countable, then implementability follows from Theorem \ref{thm:fermioniccountable}.
\end{proof}

\section{Applications}
\label{sec:applications}

\subsection{Quadratic Bosonic Interaction}
\label{subsec:QuadraticBosonic}

Our first example for a quadratic Hamiltonian whose diagonalization requires Bogoliubov transformations ``beyond the Shale/Shale--Stinespring condition'' is inspired by \cite{Derezinski2017}. We consider a free massive bosonic scalar field, which is interacting by a Wick square $ :\phi(\bx)^2: $, with $ \phi(\bx) = a^\dagger(\bx) + a(\bx) $. We discretize the momentum by putting the system in a box $ \bx \in [-\pi, \pi]^3 $ with periodic boundary conditions. Further, the Wick square is weighted by a real-valued external field $ \kappa \in C_c^\infty([-\pi, \pi]^3), \kappa(\bx) \in \CCC $. The Hamiltonian then reads
\begin{equation}
\begin{aligned}
	H = &\di \Gamma(\varepsilon_{\bp}) + \frac{1}{2}\int \kappa(\bx) :\phi(\bx)^2: \; \di \bx\\
	= &\frac{1}{2}\sum_{\bp_1, \bp_2 \in \ZZZ^3} (2 \varepsilon_{\bp_1} \delta(\bp_1 - \bp_2) a^\dagger_{\bp_1} a_{\bp_2}
	+ 2 \hat{\kappa}(- \bp_1 + \bp_2) a^\dagger_{\bp_1} a_{\bp_2} + \\
	& \qquad + \hat{\kappa}(- \bp_1 - \bp_2) a^\dagger_{\bp_1} a^\dagger_{\bp_2}
	+ \hat{\kappa}(\bp_1 + \bp_2) a_{\bp_1} a_{\bp_2} ) \;,
\end{aligned}
\end{equation}
with $ \hat{\kappa}(\bp) = \overline{\hat{\kappa}(- \bp)} $ denoting the Fourier transform of $ \kappa(\bx) $. For simplicity, we assume that $ \kappa(\bx) = \mathrm{const.} $, so we can write $ \hat{\kappa}(\bp) = \kappa \delta(\bp) $, $ \kappa \in \RRR $.

\begin{proposition}
For interactions $ \kappa > - \frac{m}{2} $ but $ \kappa \neq 0 $, the Hamiltonian $ H $ is diagonalizable in the extended sense on $ \widehat{\sH} $.
However, the transformation $ \cV $ violates the Shale condition, so $ H $ is not diagonalizable on $ \sF $.
\label{prop:applicationquadraticbosonic}
\end{proposition}

\begin{proof}
We directly compute $ \cV $ and then apply Proposition \ref{prop:diagonalizableboson}. The matrix $ A_H $ of $ A $ is
\begin{equation}
	A_H = \bigoplus_{\bp \in \ZZZ^3} A_{H,\bp} \;, \quad
	A_{H,\bp} = \begin{pmatrix} h_{\bp} & k_{\bp} \\ k_{\bp} & h_{\bp} \end{pmatrix} \;, \qquad
	h_{\bp} = (\varepsilon_{\bp} + \kappa) \;, \qquad
	k_{\bp} = \kappa \;.
\end{equation}
We diagonalize all $ A_{H,\bp} \in \CCC^{2 \times 2} $ separately via
\begin{equation}
	\cV = \bigoplus_{\bp \in \ZZZ^3} \cV_{\bp} \;, \qquad
	\cV_{\bp}^* A_{H,\bp} \cV_{\bp} = \begin{pmatrix} E_{\bp} & 0 \\ 0 & E_{\bp} \end{pmatrix} \;,
\label{eq:cVsumbosonic}
\end{equation}
with $ E_{\bp} = \sqrt{h^2 - k^2} $. Following \cite[Subsec. 1.3]{NamNapiorkowskiSolovej2016}, this is done for $ |k_{\bp}| < h_{\bp} $ by
\begin{equation}
	\cV_{\bp} = \begin{pmatrix} u_{\bp} & v_{\bp} \\ \overline{v_{\bp}} & \overline{u_{\bp}} \end{pmatrix} \;, \qquad
	u_{\bp} = c_{\bp} \;, \qquad
	v_{\bp} = c_{\bp} \frac{-G_{\bp}}{1 + \sqrt{1 - G_{\bp}^2}} \;,
\end{equation}
\begin{equation}
	G_{\bp} = k_{\bp} h_{\bp}^{-1} \;, \qquad
	c_{\bp} = \sqrt{\frac{1}{2} + \frac{1}{2 \sqrt{1 - G_{\bp}^2}}} \;.
\end{equation}
Now, $|k_{\bp}| < h_{\bp} \quad \Leftrightarrow \quad |\kappa| < \sqrt{|\bp|^2 + m^2} + \kappa $, which is satisfied for all $ \bp \in \ZZZ^3 $, if and only if $ \kappa > - \frac{m}{2} $. $ h_{\bp} > 0 $ also holds in that case and \eqref{eq:cVsumbosonic} defines a Bogoliubov transformation $ \cV $ diagonalizing $ A_H $. Regarding Proposition \ref{prop:diagonalizableboson}, $ h > 0 $ and $ \Vert h^{-1/2} k h^{-1/2} \Vert < 1 $ follow from $ h_{\bp} $ and $ |k_{\bp}| < h_{\bp} $ after taking a direct sum. Since $ v = \bigoplus_{\bp \in \ZZZ^3} v_{\bp} $ can be decomposed into modes, the same holds for $ v^*v $, which therefore has countable spectrum. So Proposition \ref{prop:diagonalizableboson} applies and $ H $ is diagonalizable in the extended sense on $ \widehat{\sH} $.\\

	It remains to show that $ \cV $ violates the Shale condition, i.e., $ \tr(v^*v) = \sum_{\bp \in \ZZZ^3} |v_{\bp}|^2 = \infty $. If $ |\bp| $ is large enough (say, $ |\bp| > p_{\max} > 0 $), we have $ \frac{\kappa d}{|\bp|} \le G_{\bp} \le \frac{\kappa}{|\bp|} $ for any $ d < 1 $, so
\begin{equation}
	|v_{\bp}|^2 = \frac{1 + \sqrt{1 - G_{\bp}^2}}{2 \sqrt{1 - G_{\bp}^2}} \frac{G_{\bp}^2}{(1 + \sqrt{1 - G_{\bp}^2})^2}
	\ge \frac{\kappa^2 d^2}{4|\bp|^2} \;.
\end{equation}
We write $ \sum_{\bp} |v_{\bp}|^2 $ as an integral, using indicator functions $ \chi_{Q(\bp)}(\cdot) $ of half-open unit cubes $ Q(\bp) $ centered at $ \bp = (p_1,p_2,p_3) $:
\begin{equation}
	\sum_{\bp \in \ZZZ^3} |v_{\bp}|^2 = \int_{\RRR^3} f(\bp') \; \di \bp' \;, \qquad
	f(\bp') = \sum_{\bp \in \ZZZ^3} |v_{\bp}|^2 \chi_{Q(\bp)}(\bp') \;.
\end{equation}
Then, for $ |\bp'| > p_{\max} $,
\begin{equation}
	\sum_{\bp \in \ZZZ^3} |v_{\bp}|^2
	\ge \int_{|\bp| > p_{\max}} f(\bp') \; \di \bp'
	\ge \int_{p_{\max}}^\infty \frac{\kappa^2 d^2}{4 (|\bp'| + \frac{\sqrt{3}}{2})^2} 4 \pi |\bp'|^2 \; \di |\bp'|
	= \infty \;,
\end{equation}
where the integral is linearly divergent, which establishes the claim.\\
\end{proof}

\begin{remark}[Infinite volume case]
Here, $ \bp \in \RRR^3 $ and we can construct an analogous $ \cV $ diagonalizing $ A_H $. However, the spectrum of $ v^*v $ is then no longer countable.\\
\end{remark}

\begin{remark}[Position-dependent $ \kappa(\bx) $]
In contrast to \cite{Derezinski2017}, we assume a constant interaction strength $ \kappa(\bx) $. Physically, it would be desirable to treat any $ \kappa \in C_c^\infty $. Then, the decomposition $ A_H = \bigoplus_{\bp \in \ZZZ} A_{H,\bp} $ fails and it might occur that $ v^*v $ has uncountable spectrum.\\
\end{remark}

\begin{remark}[Wick square is not diagonalizable]
It may be tempting to set $ \varepsilon_{\bp} = 0 $ and to try a diagonalization of only the interaction Hamiltonian $ \tfrac 12 \int \kappa(\bx) : \phi(\bx): \; \di \bx $. However, bosonic Wick squares are not diagonalizable by a Bogoliubov transformation, as ``the off-diagonal is too large''. For example, on one mode ($ \fh \cong \CCC $), the matrix associated with a Wick square $ H_W = 2 a^\dagger a + a^\dagger a^\dagger + a a $ is
\begin{equation}
	A_{H_W} = \begin{pmatrix} h & k \\ \overline{k} & \overline{h} \end{pmatrix} = \begin{pmatrix} 1 & 1 \\ 1 & 1 \end{pmatrix} \;,
\end{equation}
so $ \Vert h^{-1/2} k h^{-1/2} \Vert = 1 $ and Proposition \ref{prop:diagonalizableboson} does not apply. However, bosonic Wick products have still been constructed as self-adjoint operators by a suitable GNS construction \cite{Sanders2012}.\\
\end{remark}

\subsection{BCS Model}
\label{subsec:BCSmodel}

An example with non-implementable fermionic Bogoliubov transformations is the Bardeen--Cooper--Schrieffer (BCS) model for explaining superconductivity \cite{BardeenCooperSchrieffer1957, Haag1962}. For an overview on recent mathematical advances on BCS theory, we refer the reader to~\cite{HS16,HLR24} and the references therein. The ``Hartree-like approximation'' state \cite[(2.16)]{BardeenCooperSchrieffer1957} corresponds to a formal fermionic Bogoliubov vacuum state $ \Omega_\cV $ as in \eqref{eq:fermionicbogoliubovvacuum}. A mathematical analysis by Haag \cite{Haag1962} shows that in the infinite volume limit, the BCS Hamiltonian can indeed be diagonalized by a corresponding Bogoliubov transformation, which is not implementable on Fock space.\\
We consider a similar model of a fermionic gas inside a box with periodic boundary conditions $ \bx \in [-\pi, \pi]^3 $. Hence, we have discretized momenta $ \bp \in \ZZZ^3 $, as well as two spins $ s \in \{ \uparrow, \downarrow \} $, leading to a one-particle Hilbert space $ \fh = L^2(\ZZZ^3 \times \{ \uparrow, \downarrow \}) $. The corresponding Fock space is $ \sF = \sF(\ZZZ^3 \times \{ \uparrow, \downarrow \} ) $. We consider the following quadratic Hamiltonian (see \cite{Haag1962}), which provides an approximate description for the fermionic gas that becomes exact in the infinite volume limit:
\begin{equation}
	H' = H_0 + H_I' = \sum_{\substack{\bp \in \ZZZ^3 }} \left( \varepsilon_{\bp} a_{\bp,\uparrow}^\dagger a_{\bp,\uparrow} + \varepsilon_{\bp} a_{\bp,\downarrow}^\dagger a_{\bp,\downarrow} - \tilde{\Delta}_{\bp} a_{\bp,\uparrow}^\dagger a_{\bp,\downarrow}^\dagger + \overline{\tilde{\Delta}_{\bp}} a_{\bp,\uparrow} a_{\bp,\downarrow} \right) \;,
\label{eq:hamiltonianBCS}
\end{equation}
with kinetic energy $ \varepsilon_{\bp} = \frac{\bp^2}{2m} - \mu \in \RRR $ and interaction strength $ \tilde{\Delta}_{\bp} \in \CCC $, of which we assume $ \tilde{\Delta}_{\bp} \neq 0 $. As a basis $ (e_j)_{j \in J} $ for identifying $ \fh $ with $ \ell^2 $, we choose
\begin{equation}
	(e_{\bp,s})_{\substack{\bp \in \ZZZ^3 \\ s \in \{ \uparrow, \downarrow \} }} \subset L^2(\ZZZ^3 \times \{ \uparrow, \downarrow \}) \;,
	\qquad e_{\bp,s}(\bp',s') = \delta_{\bp \bp'} \delta_{ss'} \;,
\end{equation}
with $ \delta $ being the Kronecker delta. The corresponding canonical basis of $ \ell^2 $ is denoted $ (\be_{\bp,s})_{\bp \in \ZZZ^3, s \in \{ \uparrow, \downarrow \} } $. In order to obtain momentum conservation, we have to interpret $ a^\dagger_{\bp,\downarrow}, a_{\bp,\downarrow} $ as creating/annihilating a fermion of momentum $ - \bp $. The mode index $ \bp $ is only used for an easier decomposition into modes.\\

\begin{proposition}
	The Hamiltonian $ H' $ \eqref{eq:hamiltonianBCS} is diagonalizable in the extended sense on $ \widehat{\sH} $.
\label{prop:applicationBCS}
\end{proposition}

\begin{proof}
We compute $ \cV $ directly and apply Proposition \ref{prop:diagonalizablefermion}.
This is done in the block matrix representation $ A_{H'} = \bigoplus_{\bp} A_{H',\bp} $ and $ \cV = \bigoplus_{\bp} \cV_{\bp} $ with
\begin{equation}
	A_{H',\bp} = \begin{pmatrix}
		\varepsilon_{\bp} & 0 & 0 & -\tilde{\Delta}_{\bp} \\
		0 & \varepsilon_{\bp} & \tilde{\Delta}_{\bp} & 0 \\
		0 & \overline{\tilde{\Delta}_{\bp}} & -\varepsilon_{\bp} & 0 \\
		-\overline{\tilde{\Delta}_{\bp}} & 0 & 0 & -\varepsilon_{\bp} \\
	\end{pmatrix} \;, \qquad
	\cV_{\bp} = \begin{pmatrix}
		u_{\bp} & 0 & 0 & v_{\bp} \\
		0 & u_{\bp} & -v_{\bp} & 0 \\
		0 & \overline{v}_{\bp} & \overline{u}_{\bp} & 0 \\
		-\overline{v}_{\bp} & 0 & 0 & \overline{u}_{\bp} \\
	\end{pmatrix} \;,
\label{eq:AH'bp}
\end{equation}
where $ A_{H',\bp}, \cV_{\bp,s} \in \CCC^4 \otimes \CCC^4 $. The diagonalized matrix reads
\begin{equation}
	\cV_{\bp}^* A_{H',\bp} \cV_{\bp} = \begin{pmatrix}
		E_{\bp} & 0 & 0 & 0 \\
		0 & E_{\bp} & 0 & 0 \\
		0 & 0 & -E_{\bp} & 0 \\
		0 & 0 & 0 & -E_{\bp} \\
	\end{pmatrix} \;, \qquad
	E_{\bp} = \sqrt{\varepsilon_{\bp}^2 + |\tilde{\Delta}_{\bp}|^2} \;,
\end{equation}
and the diagonalization is established by
\begin{equation}
	u_{\bp} = \frac{\tilde{\Delta}_{\bp}}{\sqrt{(E_{\bp} - \varepsilon_{\bp})^2 + |\tilde{\Delta}_{\bp}|^2 }} \;, \qquad
	v_{\bp} = \frac{E_{\bp} - \varepsilon_{\bp}}{\sqrt{(E_{\bp} - \varepsilon_{\bp})^2 + |\tilde{\Delta}_{\bp}|^2 }} \;.
\label{eq:upvp}
\end{equation}

From this form, it also follows that $ \mathrm{dimKer}(A_H) $ is either even or $ \infty $. So in order to apply Proposition \ref{prop:diagonalizablefermion}, we only need to show that the spectrum of $ v^*v $ is countable. This is the case, since $ v = \bigoplus_{\bp \in \ZZZ^3} v_{\bp} $ decays into modes, so also $ v^*v = \bigoplus_{\bp \in \ZZZ^3} (v^*v)_{\bp} $ decays into modes, where each $ (v^*v)_{\bp} $ is a finite-dimensional matrix with finite spectrum. As the sum over $ \bp $ is countable, also the spectrum of $ v^*v $ is countable, and by Proposition \ref{prop:diagonalizablefermion}, $ H $ is diagonalizable on $ \widehat{\sH} $.\\
\end{proof}

\begin{remark}[Infinite volume case]
The Hamiltonian $ H' $ is an approximation to $ H = H_0 + H_I $, where $ H_I $ is an attractive quartic interaction between fermion pairs. As mentioned above, this approximation is only exact in the infinite volume limit. In that case, $ \bp \in \RRR^3 $ and \eqref{eq:AH'bp}--\eqref{eq:upvp} still yield a Bogoliubov transformation $ \cV $ diagonalizing $ A_{H'} $. However, $ v^*v $ has generally uncountable spectrum, so Theorem \ref{thm:fermioniccountable} does not apply. \\
\end{remark}

\begin{appendices}
\addtocontents{toc}{\setcounter{tocdepth}{-2}}

\section{Results on the Infinite Tensor Product Space}
\label{sec:resultsITP}

Our first proposition characterizes when exactly two $ C $-sequences correspond to the same functional $ \Psi \in \widehat{\sH} = \prod_{k \in I}^\otimes \sH_k $. This allows us to define operators on $ \widehat{\sH} $ in Section \ref{subsec:extopalg}. Recall that any $ C $-sequence $ (\Psi) = (\Psi_k)_{k \in I} $ gives rise to a unique functional $ \Psi \in \widehat{\sH} $ by the embedding $ \Psi = \iota((\Psi)) $.

\begin{proposition}
	Whenever $ (\Psi), (\Psi') \in \Cseq $ represent the same functional $ \Psi = \Psi' $, then there exists a family of complex numbers $ (c_k)_{k \in I} $ such that
\begin{equation}
	\Psi_k = c_k \Psi_k' \quad \forall k \in I \;, \qquad \textnormal{and} \qquad
	\prod_{k \in I} c_k = 1 \;,
\label{eq:cksequence}
\end{equation}
using the notion of convergence for an infinite product from Section \ref{subsec:infinitetensorprod}.\\
Conversely, if $ (\Psi), (\Psi') \in \Cseq $ just differ by a family $ (c_k)_{k \in I} $ as in \eqref{eq:cksequence}, then they represent the same functional $ \Psi = \Psi' $.\\
\label{prop:cksequence}
\end{proposition}
\begin{proof}
	For the first statement, we must prove that $ \Psi_k $ and $ \Psi_k' $ are parallel for any $ k \in I $. So let us fix a $ k $ and decompose $ \Psi_k = \Psi_k^\parallel + \Psi_k^\perp $ with $ \Psi_k^\parallel \parallel \Psi_k' $ and $ \Psi_k^\perp \perp \Psi_k' $, and suppose that $ \Psi_k^\perp \neq 0 $. Now, choose some $ C $-sequences $ (\Phi), (\Phi') $, that agree on all $ k' \neq k $ and with $ \Phi_k \parallel \Psi_k $ and $ \Phi'_k \parallel \Psi'_k $, as well as $ \Vert \Phi_k \Vert_k = \Vert \Phi'_k \Vert_k = 1 $. Then,
\begin{equation}
\begin{aligned}
	\langle \Psi, \Phi' \rangle
	= &\langle \Psi_k, \Phi_k' \rangle_k \prod_{k \neq k'} \langle \Psi_{k'}, \Phi'_{k'} \rangle_{k'} 
	= \Vert \Psi_k^\parallel \Vert_k \prod_{k \neq k'} \langle \Psi_{k'}, \Phi_{k'} \rangle_{k'}\\
	< &\Vert \Psi_k \Vert_k \prod_{k \neq k'} \langle \Psi_{k'}, \Phi_{k'} \rangle_{k'}
	= \langle \Psi_k, \Phi_k \rangle_k \prod_{k \neq k'} \langle \Psi_{k'}, \Phi_{k'} \rangle_{k'}
	= \langle \Psi, \Phi \rangle \;.
\end{aligned}
\label{eq:smallerinequality}
\end{equation}
By the same arguments, $ \langle \Psi', \Phi \rangle < \langle \Psi', \Phi' \rangle $. But since $ (\Psi), (\Psi') $ correspond to the same functional $ \Psi = \Psi' $, we can freely exchange both expressions within the scalar product:
\begin{equation}
	\langle \Psi, \Phi' \rangle
	= \langle \Psi', \Phi' \rangle
	> \langle \Psi', \Phi \rangle
	= \langle \Psi, \Phi \rangle \;.
\label{eq:largerinequality}
\end{equation}
This contradicts \eqref{eq:smallerinequality} and thus establishes $ \Psi_k = c_k \Psi_k' $. Convergence of $ \prod_{k \in I} c_k $ can be seen as follows: We have
\begin{equation}
	\Vert \Psi \Vert^2 = \langle \Psi', \Psi \rangle
	= \prod_{k \in I} \langle \Psi'_k, c_k \Psi'_k \rangle_k
	= \prod_{k \in I} c_k \Vert \Psi'_k \Vert_k^2 \;.
\label{eq:prodconvergence}
\end{equation}
So if $ \prod_{k \in I} c_k $ was not convergent, i.e., $ \sum_k |c_k - 1| = \infty $, then for the product on the right-hand side, we would have
\begin{equation}
	 \sum_k \left\vert c_k \Vert \Psi'_k \Vert_k^2 - 1 \right\vert
	 \ge \underbrace{\sum_k \Vert \Psi'_k \Vert_k^2 \left\vert c_k - 1 \right\vert}_{(*)} - \underbrace{\sum_k \left\vert \Vert \Psi'_k \Vert_k^2 - 1 \right\vert }_{< \infty} \;.
\label{eq:prodconvergence2}
\end{equation}
Now, $ \Vert \Psi'_k \Vert_k^2 > 1/2 $ for all but finitely many $ k $, so $ (*) $ and thus the first expression in \eqref{eq:prodconvergence2} diverges. This is a contradiction to \eqref{eq:prodconvergence} being convergent. So $ \prod_{k \in I} c_k $ indeed yields a complex number.\\
But since $ \Vert \Psi \Vert^2 = \Vert \Psi' \Vert^2 = \prod_{k \in I} \Vert \Psi'_k \Vert_k^2 $, we immediately obtain $ \prod_{k \in I} c_k = 1 $ from \eqref{eq:prodconvergence}.\\

The converse statement can readily be seen by computing the action of the functionals $ \Psi, \Psi' $ on some $ \Phi \in \Cseq $:
\begin{equation}
	\langle \Psi, \Phi \rangle
	= \prod_{k \in I} \langle \Psi_k, \Phi_k \rangle_k
	= \prod_{k \in I} \langle c_k \Psi'_k, \Phi_k \rangle_k
	= \left( \prod_{k \in I} \overline{c_k} \right) \prod_{k \in I} \langle \Psi'_k, \Phi_k \rangle_k
	= \langle \Psi', \Phi \rangle \;.
\end{equation}
\end{proof}

By \cite[Lemma 4.1.1]{vonNeumann1939}, all subspaces $ \prod^{\otimes C}_{k \in I} \sH_k $ of $ \widehat{\sH} = \prod^\otimes_{k \in I} \sH_k $ are \textbf{mutually orthogonal}. This allows for a particularly simple decomposition.

\begin{lemma}
	For any $ \Psi \in \widehat{\sH} $, we can write
\begin{equation}
	\Psi
	= \sum_{m \in \cM} d_m \Psi^{(m)}
	= \sum_{m \in \cM} d_m \prod^\otimes_{k \in I} \Psi^{(m)}_k \;,
\label{eq:psilinearcombination}
\end{equation}
with $ \cM $ being a subset of $ \NNN $, $ \Psi^{(m)} $ defined by the mutually orthogonal $ C_0 $-sequences $ (\Psi^{(m)}) $ with $ \Vert \Psi^{(m)} \Vert = 1 $ and where $ \sum_m |d_m|^2 < \infty $ is a complex sequence.\\
	Moreover, one can choose a fixed set $ Z = \{\Psi^{(a)}\}_{a \in A} $ defined by mutually orthogonal, normalized $ C_0 $-sequences $ (\Psi^{(a)}) $, such that for all $ \Psi \in \widehat{\sH} $, the form \eqref{eq:psilinearcombination} can be achieved by taking only $ \Psi^{(m)} \in Z $. The decomposition \eqref{eq:psilinearcombination} is then unique up to the choice of the $ \Psi^{(m)}_k $ representing $ \Psi^{(m)} $.\\
\label{lem:psilinearcombination}
\end{lemma}

So $ Z $ is an orthonormal basis of $ \widehat{\sH} $ that might be uncountable, but the elements $ \Psi \in \widehat{\sH} $ are all countable linear combinations with coefficient sequences in $ \ell^2 $.\\

\begin{proof}
	By definition, any $ \Psi \in \widehat{\sH} $ can be approximated by a Cauchy sequence $ (\Psi^{(r)})_{r \in \NNN} \subseteq \widetilde{\prod}^{\otimes}_{k \in I} \sH_k, \Psi^{(r)} \to \Psi $. So each $ \Psi^{(r)} $ can be written as a finite linear combination of $ C $-sequences. All $ C $-sequences, that are no $ C_0 $-sequences, must have norm 0, so we drop them and simply write
\begin{equation}
	\Psi^{(r)} = \sum_{\ell = 1}^{L_r} \Psi^{(r)}_{\ell} \;,
\end{equation}
with $ (\Psi^{(r)}_{\ell}) $ being $ C_0 $-sequences. Now, the $ C_0 $-sequences decay into mutually orthogonal equivalence classes $ C $, out of which countably many are occupied by any $ \Psi^{(r)} $. So
\begin{equation}
	\Psi^{(r)} = \sum_C \sum_{\ell :(\Psi^{(r)}_{\ell}) \in C} \Psi^{(r)}_{\ell} =: \sum_{C} \Psi^{(r)}_C \;.
\label{eq:PsirC}
\end{equation}
By orthogonality of the subspaces $ \Psi \in \prod^{\otimes C}_{k \in I} \sH_k $, we have
\begin{equation}
	\Vert \Psi^{(r)} - \Psi^{(s)} \Vert^2 = \sum_C \Vert \Psi^{(r)}_C - \Psi^{(s)}_C \Vert^2 \;.
\end{equation}
$ (\Psi^{(r)})_{r \in \NNN} $ is a Cauchy sequence, so $ (\Psi^{(r)}_C)_{r \in \NNN} $ is also a Cauchy sequence for all $ C $. That means, the limit $ \Psi_C = \lim_{r \to \infty} \Psi^{(r)}_C $ exists and by orthogonality of the $ \Psi^{(r)}_C $ for each $ r $,
\begin{equation}
	\lim_{r \to \infty} \sum_C \Psi^{(r)}_C = \sum_C \lim_{r \to \infty} \Psi^{(r)}_C \qquad \Leftrightarrow \qquad \Psi = \sum_C \Psi_C \;.
\end{equation}
We may now write $ \Psi_C $ in coordinates:
\begin{equation}
	\Psi_C
	= \lim_{r \to \infty} \Psi^{(r)}_C
	= \lim_{r \to \infty} \sum_{\ell :(\Psi^{(r)}_{\ell}) \in C} \Psi^{(r)}_{\ell}
	= \sum_{n(\cdot) \in F} a_C(n(\cdot)) \prod^\otimes_{k \in I} e_{k,n(k)} \;,
\end{equation}
so also $ \Psi $ can be written as a countable sum over mutually orthogonal, normalized $ C_0 $-sequences with coordinates $ a_C(n(\cdot)) $. We index the sequences and coordinates by $ (\Psi^{(m)}) $ and $ d_m $, which yields the desired form \eqref{eq:psilinearcombination}.\\
Square summability of the $ d_m $ can be seen by
\begin{equation}
	\sum_m |d_m|^2 = \sum_C \sum_{n(\cdot) \in F} |a_C(n(\cdot))|^2 \;.
\end{equation}

Now, the set $ Z = \{ \Psi^{(a)} \}_{a \in A} $ is exactly the union of all vectors $ \prod^\otimes_{k \in I} e_{k,n(k)} $ over all classes $ C $, which is indeed a set of mutually orthogonal, normalized vectors.\\
Uniqueness of the decomposition follows by orthogonality of the $ \Psi^{(r)} $.\\
\end{proof}

Further, by \cite[Def. 6.1.1]{vonNeumann1939}, two $ C_0 $-sequences $ (\Psi), (\Phi) $ are weakly equivalent, if and only if there exists a family $ (z_k)_{k \in I} \subseteq \CCC $ with $ (z_k \Psi_k)_{k \in I} $ being (strongly) equivalent to $ (\Phi_k)_{k \in I} $. From that, we conclude:

\begin{lemma}
Let $ C_w $ be the weak equivalence class of a $ C_0 $-sequence $ (\Phi) = (\Phi_k)_{k \in I} $, choose an orthonormal basis $ (e_{k,n})_{n \in \NNN_0} $ for each $ \sH_k $ such that $ \Phi_k = c_{k,0} $ and define for any $ C_0 $-sequence $ (\Psi) = (\Psi_k)_{k \in I} $ the coordinates $ c_{k,n} := \langle e_{k,n}, \Psi_k \rangle_k $.\\
Then, $ \prod^{\otimes C_w}_{k \in I} \sH_k $ is exactly the closure of the span of all normalized $ C_0 $-sequences, where $ |c_{k,0}| = 1 $ for all but finitely many $ k \in I $.
\label{lem:weakequivalence}
\end{lemma}

Note that the last statement means $ c_{k,n} = 0 $ for $ n \ge 1 $ and for those $ k $. In simple words, Lemma \ref{lem:weakequivalence} asserts that replacing $ C $ by $ C_w $ in the equivalence class is done via replacing $ c_{k,0} = 1 $ by $ |c_{k,0}| = 1 $.\\

\begin{proof}
First, we prove that any normalized $ C_0 $-sequence $ (\Psi) $ with $ |c_{k,0}| = 1 $ almost everywhere is weakly equivalent to $ (\Phi) $: We can define a family of phase rotations $ |z_k| = 1 $, such that $ z_k c_{k,0} = 1 $ for all $ k $ with $ |c_{k,0}| = 1 $. So $ (z_k \Psi_k)_{k \in I} $ has $ z_k c_{k,0} = 1 $ almost everywhere and is hence a $ C_0 $-sequence strongly equivalent to $ (\Phi) $. Therefore, $ (\Psi) $ is weakly equivalent to $ (\Phi) $. So $ \Psi \in \prod^{\otimes C_w}_{k \in I} \sH_k $ and the same holds for the span of these $ C_0 $-sequences and their closure with respect to the Hilbert space topology on $ \widehat{\sH} $.\\
Conversely, any $ \Psi \in \prod^{\otimes C_w}_{k \in I} \sH_k $ is within the closure of the span of normalized $ C_0 $-sequences with $ |c_{k,0}| = 1 $ almost everywhere: By Lemma \ref{lem:psilinearcombination}, we may write
\begin{equation*}
	\Psi = \sum_{m \in \cM} d_m \prod^\otimes_{k \in I} \Psi^{(m)}_k \;,
\end{equation*}
where $ \sum_m |d_m|^2 < \infty $ and the $ (\Psi^{(m)}) = (\Psi^{(m)}_k)_{k \in I} $ with $ \Vert \Psi^{(m)} \Vert = 1 $ are orthogonal. Further, we may choose $ (\Psi^{(m)}) \sim_w (\Phi) $, since all $ (\Psi^{(m)}) $ were constructed to come from a (strong) equivalence class $ C $ contained within $ C_w $. So there exist families $ (z^{(m)}_k)_{k \in I}, \; |z^{(m)}_k| = 1 $, such that $ (z^{(m)}_k \Psi^{(m)}_k)_{k \in I} \sim (\Phi) $ for all $ m \in \cM $. By strong equivalence, we may approximate each $ (z^{(m)}_k \Psi^{(m)}_k) $ up to arbitrary precision $ \varepsilon > 0 $ by a linear combination of (normalized) families $ \Psi^{(m,\varepsilon)} $, such that, when writing these families in coordinates, we have $ c^{(m,\varepsilon)}_{k,0} = 1 $ almost everywhere in $ k \in I $. Hence, the families $ ((z^{(m)}_k)^{-1} \Psi^{(m,\varepsilon)}_k)_{k \in I} $ approximate $ \Psi^{(m)} $ up to precision $ \varepsilon $. They satisfy $ |(z^{(m)}_k)^{-1} c^{(m,\varepsilon)}_{k,0}| = 1 $ almost everywhere, so $ \Psi^{(m)} $ can be approximated up to arbitrary precision by a linear combination of $ C_0 $-sequences with the above-mentioned property. And by \eqref{eq:psilinearcombination} and convergence of $ |d_m|^2 $, also an arbitrary approximation of $ \Psi $ is possible by linear combinations of normalized $ C_0 $-sequences with $ |c_{k,0}| = 1 $ almost everywhere.\\

\end{proof}

\section{Extended State Space}
\label{app:ESS}

As mentioned in the introduction, it is also possible to achieve results similar to Theorems~\ref{thm:bosoniccountable} and \ref{thm:fermioniccountable} in the ``Extended State Space'' setting from~\cite{lill,lillproc}. In this appendix, we define extensions $ \FB, \FB_{\ex} $ adapted to $ v^* v $ having discrete spectrum, so the definitions are simpler than in~\cite{lill,lillproc}.

\subsection{ESS Construction}
\label{subsec:extendedstatespace}

We start by defining
\begin{itemize}
\item The space of \textbf{generalized one-particle wave functions}
\begin{equation}
	\cE = \cE(\NNN) = \CCC^\NNN := \{ \bphi: \NNN \to \CCC \} \;,
\label{eq:cEej}
\end{equation}
that is, the space of complex sequences $ \bphi = (\phi_j)_{j \in \NNN} $. It extends $ \cD_{\bg} $ (bosonic) or $ \cD_{\boeta} $ (fermionic) and replaces the smooth function space $ \Sdot_1 $ from~\cite{lill}.

\item The space of \textbf{generalized $ N $-particle wave functions}
\begin{equation}
	\cE^{(N)}(\NNN) := \{ \Psi: \NNN^N \to \CCC \} \;.
\label{eq:cENA}
\end{equation}

\item The space of \textbf{generalized Fock space functions} (which replaces $ \Sdot_\sF $ from \cite{lill}):
\begin{equation}
	\cE_\sF(\NNN) := \bigoplus_{N = 0}^\infty \cE^{(N)}(\NNN) = \{ \Psi : \cQ(\NNN) \to \CCC \} \;.
\end{equation}
\end{itemize}

$ \cE_\sF $ is \textit{not} yet the final extended state space. Formal state vectors occurring in QFT include products of functions $ \Psi_m \in \cE_\sF $ and exponentials of divergent sums $ e^\fr $ with ``$ \fr = \pm \infty $'', see \cite{lill}. Hence, we need to construct structures accommodating such infinite quantities.\\
First, we introduce a \textbf{space of renormalization factors} $ \Ren_1(\NNN) $, whose elements are sequences that represent formal (and possibly divergent) series, thus replacing the formal (and possibly divergent) integrals from $ \Ren_1 $ in \cite{lill}.

\begin{itemize}
\item $ \Ren_1(\NNN) := \cE /_{\sim_{\Ren_1}} $. Here, for $ \br_1, \br_2 \in \cE(\NNN) $, we define $ \br_1 \sim_{\Ren_1} \br_2 $ if and only if $ (\br_1 - \br_2) \in \ell^1 = L^1(\NNN) $ and $ \sum_{j \in \NNN} (r_{1,j} - r_{2,j}) = 0 $.
\end{itemize}
We denote elements of $ \Ren_1 $ by $ \fr = [\br] $ and identify $ \fr = \sum_{j \in \NNN} r_j \in \CCC $, if $ \br \in \ell^1 $. All $ \Ren_1 $-elements not identifiable with a $ \CCC $-number can be thought of as ``controlled infinitely large numbers''.\\
Multiplication of $ \Ren_1 $-elements is allowed by introducing the free vector spaces $ \Pol_P(\NNN) $, $ P \in \NNN $, called space of \textbf{renormalization polynomials of degree $ P $}, and defined as being spanned by all commutative products $ \fr_1 \cdot \ldots \cdot \fr_p, \; p \le P, \; \fr_j \in \Ren_1(\NNN) $. Again, we identify equivalent terms by modding out an equivalence relation:
\begin{itemize}
\item $ \Ren_P(\NNN) := \Pol_P(\NNN) /_{\sim_{\Ren_P}} $ with $ \sim_{\Ren_P} $ generated by both $ \fr_1 \fr_2 \ldots \fr_p \sim_{\Ren_P} c_1 \fr_2 \ldots \fr_p $ for $ \fr_1 = \sum_{j \in \NNN} r_{1,j} = c_1 \in \CCC $ and $ (c_1 c_2) \fr_1\ldots \fr_p \sim_{\Ren_P} c_1 (c_2 \fr_1) \ldots \fr_p $.
\end{itemize}
The space of \textbf{renormalization polynomials} is then given by
\begin{equation}
	\Ren(\NNN) := \bigcup_{P \in \NNN} \Ren_P(\NNN) \;.
\end{equation}

Wave function renormalization factors usually take the form $ e^\fr $, where $ \fr \in \Ren_1 $ may or may not correspond to a finite number, as well as linear combinations of such terms $ e^\fr $. We will interpret them as elements of a suitably defined field $ \eRen $. For this, consider the group algebra $ \CCC[\Ren_1] $, where the group is $ (\Ren_1, +) $ with addition as described above. We identify elements $ \fr \in \Ren_1 $ with the symbolic expression $ e^{\fr} $, both to distinguish the two additions in $ (\Ren_1, +) $ and $ \CCC[\Ren_1] $, as well as to be congruent with the intuition that $ e^{\fr} $ is an exponential of a finite or infinite number. Elements in $ \CCC[\Ren_1] $ are then denoted as $ c_1 e^{\fr_1} + \ldots + c_M e^{\fr_M}, \; c_j \in \CCC, \; \fr_j \in \Ren_1 $. Some elements of $ \CCC[\Ren_1] $ are intuitively equal to zero. We set them equivalent to zero by modding out an ideal $ \cI \subset \CCC[\Ren_1] $ generated by all elements $ e^c e^\fr - e^{c + \fr} $ with $ c \in \CCC, \; \fr \in \Ren_1 $.\\

As in \cite[Proposition~3.2]{lill} one proves that the quotient ring $ \CCC[\Ren_1]/\cI $ has no proper zero divisors. So the following quotient field exists, which is a field extension of $ \CCC $:
\begin{itemize}
\item $ \eRen(\NNN) := \{ \fc = a_1 / a_2 \; \mid \; a_1, a_2 \in \CCC[\Ren_1(\NNN)]/\cI \} $ is the \textbf{field of wave function renormalizations}.
\end{itemize}

The formal state vectors appearing in QFT are of the form $ \Psi = \sum_m \fc_m \Psi_m $ with $ \fc_m \in \eRen $ and $ \Psi_m \in \cE_\sF $. Those can be described as elements of the following space:
\begin{itemize}
\item $ \FB(\NNN) := \FB_0(\NNN) /_{\sim_{\rm F}} $ is the \textbf{first extended state space}, where $ \FB_0(\NNN) $ is the free $ \eRen $-vector space over $ \cE_\sF $ (all finite $ \eRen $-linear combinations) and $ \sim_{\rm F} $ is generated by $ (c \fc) \Psi_m \sim_{\rm F} \fc (c \Psi_m) $ for $ c \in \CCC $.
\end{itemize}

For intermediate calculations, we will need an even larger vector space $ \FB_{\ex} $, which allows for multiplication of $ \Psi $ by elements of $ \Ren $. We define $ \Ren^\cQ(\NNN) $ to consist of all functions $ \cQ(\NNN) \to \Ren $, which replaces the similarly-defined $ \Ren^{\Qdot} $ from \cite{lill}.\\

\begin{itemize}
\item $ \FB_{\ex}(\NNN) := \FB_{\ex,0}(\NNN) /_{\sim_{\rm Fex}} $ is the \textbf{second extended state space}, where $ \FB_{\ex,0}(\NNN) $ is the set of all countable $ \eRen $-linear combinations $ \Psi = \sum_{m \in \NNN} \fc_m \Psi_m $ with $ \fc_m \in \eRen(\NNN), \Psi_m \in \Ren^\cQ(\NNN) $ and where $ \sim_{\rm Fex} $ is generated by $ (c \fc) \Psi_m \sim_{\rm Fex} \fc (c \Psi_m) $ for $ c \in \CCC $.
\end{itemize}

We may embed the complex numbers into $ \Ren $ by identifying $ c \in \CCC $ with $ c e^0 \in \eRen $. Hence, the space $ \FB $ can be embedded into $ \FB_{\ex} $.\\

\begin{remark}[Comparison with non-standard analysis] \label{rem:nonstandard_analysis}
The construction of $ \Ren_1 $ might remind about that of the extended real line $ ^*\RRR $ in non-standard analysis~\cite[Sect.~1.1]{Albeverio1986}, since both extend $ \RRR $ by modding out an equivalence relation on sequences. One may therefore ask whether $ \Ren_1 $ can be naturally embedded into $ ^*\RRR $ or vice versa. This is not the case. First, note that $ \Ren_1 $ extends $ \CCC $ and not just $ \RRR $, while there are no elements corresponding to complex numbers in $ ^*\RRR $. Even when restricting to real sequences, an embedding would require to identify a sequence $ \br = (r_j)_{j \in \NNN} $, $ [\br] \in \Ren_1 $, with its series $ \bs(\br) := (s(\br)_j)_{j \in \NNN} $, $ s(\br)_j := \sum_{\ell = 1}^j r_\ell $ with respect to $ ^*\RRR $. This is needed to be consistent with the identification $ \br = c \in \RRR $ in $ \Ren_1 $ whenever $ \sum_j r_j = c < \infty $, as well as the identification $ \bs(\br) = c \in \RRR $ in $ ^*\RRR $, whenever $ s(\br)_j = c $ for all $ j > J $, $ J \in \NNN $.\\
Now, there always exists an infinite set $ A \subset \NNN $, such that $ \bs(\br_1) $ and $ \bs(\br_2) $ are equivalent with respect to $ ^*\RRR $ whenever they agree on $ \NNN \setminus A $. However, the difference between $ \br_1 $ and $ \br_2 $ on $ A $ may not be $ \ell^1 $-summable, rendering the two sequences inequivalent with respect to $ \Ren_1 $. So $ \Ren_1 $ cannot be naturally embedded into $ ^*\RRR $.\\
Conversely, $ ^*\RRR $ contains many inequivalent $ \bs(\br) $ converging to the same $ c \in \RRR $ (corresponding to infinitesimals), even with $ \ell^1 $-summable differences, while the corresponding $ \br $ are all equivalent with respect to $ \Ren_1 $.\\
\end{remark}

\subsection{Operator Lift and Implementation}
\label{subsec:finalESS}

Creation and annihilation operators $ a^\dagger(\bphi), a(\bphi) $ are defined on $ \Psi_m: \cQ(\NNN) \to \CCC $ in similarity to \eqref{eq:aadagger}. Formally,
\begin{equation}
	(a_\pm^\dagger(\bphi) \Psi_m)(q) = \sum_{k = 1}^N \frac{(\pm 1)^k}{\sqrt{N}} \phi_{j_k} \Psi_m(q \setminus j_k) \;, \qquad
	(a_\pm(\bphi) \Psi_m)(q) = \sqrt{N+1} \sum_j \overline{\phi_j} \Psi_m(q,j) \;.
\label{eq:aadaggerESS}
\end{equation}
Using that there is a natural distribution pairing between $ \cD $ and $ \cE $, one easily verifies the following well-definedness statement.
\begin{lemma}[Products of $ a^\dagger, a $ are well-defined on the ESS]\ \\
	Consider the ESS $ \FB $ built over $ \cE $. Then, \eqref{eq:aadaggerESS} uniquely defines operators $ a^\dagger(\bphi), a(\bphi) $ as follows:
\begin{equation}
\begin{aligned}
	&a^\dagger(\bphi): \FB \to \FB \;, \qquad
	&&a(\bphi): \FB \to \FB 
	&&\qquad \forall \bphi \in \cD \;, \\
	&a^\dagger(\bphi): \FB \to \FB \;, \qquad
	&&a(\bphi): \FB \to \FB_{\ex}
	&&\qquad \forall \bphi \in \cE \;.
\end{aligned}
\end{equation}
\label{lem:aadaggerexistESS}
\end{lemma}
Note that the CAR/CCR are a direct consequence of \eqref{eq:aadaggerESS} and hence still valid for the operator extensions.\\
The condition for $ \UUU_\cV: \cD_\sF \to \FB $ being an ``extended implementer'' is identical to the ITP case (Definition \ref{def:implementation}). Also Definition \ref{def:transformbogoliubovstate} for $ \UUU_\cV $, given a Bogoliubov vacuum $ \Omega_\cV \in \FB $, carries over from ITP to ESS. Then, Lemma \ref{lem:UUUcV} (well-definedness of $ \UUU_\cV $) and Lemma \ref{lem:implementation} are established analogously. By checking the four conditions as in the proofs of Theorems \ref{thm:bosoniccountable} and \ref{thm:fermioniccountable}, it is then straightforward to verify the following two results.

\begin{proposition}[Implementation via ESS works, bosonic]
	 Consider a bosonic Bogoliubov transformation $ \cV = \left( \begin{smallmatrix} u & v \\ \overline{v} & \overline{u} \end{smallmatrix} \right) $ with $ v^*v $ having countable spectrum. Let $ \FB $ be the ESS over $ \cE $ with respect to the basis $ (\bg_k)_{k \in \NNN} \subset \ell^2 $. Define the new vacuum vector
\begin{equation}
	\Omega_\cV = \underbrace{ \exp\left( \frac{1}{4} \sum_k \log\left(1 - \tfrac{\nu_k^2}{\mu_k^2} \right) \right) }_{=: e^\fr} \underbrace{ \exp{\left( -\sum_k \frac{\nu_k}{2 \mu_k} (a^\dagger(\bg_k))^2 \right) } }_{=: \Psi_\cV} \Omega = e^\fr \Psi_\cV \;,
\label{eq:bosonicbogoliubovvacuum2}
\end{equation}
	where $ \mu_k, \nu_k $ are the singular values of $ u, v $ as in Section \ref{subsec:implementation1}. Then, $ \cV $ is implemented in the sense of \eqref{eq:implementation} by $ \UUU_\cV: \cD_\sF \to \FB $ \eqref{eq:transformbogoliubovstate}.\\
\label{prop:bosoniccountableESS}
\end{proposition}

\begin{proposition}[Implementation via ESS works, fermionic]
	 Consider a fermionic Bogoliubov transformation $ \cV = \left( \begin{smallmatrix} u & v \\ \overline{v} & \overline{u} \end{smallmatrix} \right) $ with $ v^*v $ having countable spectrum. Let $ \FB $ be the ESS over $ \cE $ with respect to the basis $ (\boeta_j)_{j \in J} $, and let $ |J''_1| < \infty $, so \textbf{the number of modes with a full particle--hole transformation is finite}. Define the new vacuum vector
\begin{equation}
\begin{aligned}
	\Omega_\cV 
	= \underbrace{ \exp \left(\sum_{i \in I'} \log \alpha_i \right) }_{=: e^\fr} \underbrace{ \left( \prod_{j \in J''_1}^\otimes a^\dagger(\boeta_j) \right) \left( \prod_{i \in I'}^\otimes \left( 1 - \frac{\beta_i}{\alpha_i} a^\dagger(\boeta_{2i}) a^\dagger(\boeta_{2i-1}) \right) \right) }_{=: \Psi_\cV} \Omega = e^\fr \Psi_\cV \;,
\end{aligned}
\label{eq:fermionicbogoliubovvacuum2}
\end{equation}
	with $ \alpha_i, \beta_i $ being the singular values of $ u, v $ as in \eqref{eq:fermionictrafo3}. Then, $ \cV $ is implemented in the sense of \eqref{eq:implementation} by $ \UUU_\cV: \cD_\sF \to \FB $ \eqref{eq:transformbogoliubovstate}.\\
\label{prop:fermioniccountableESS}
\end{proposition}

In contrast to Theorem \ref{thm:fermioniccountable}, the additional condition of having finitely many modes with a particle--hole transformation comes from the requirement that $ \Omega_\cV $ be in $ \FB $.\\
Diagonalizability of Hamiltonians in an extended sense can then be defined for $ \FB $ exactly as for $ \widehat{\sH} $ (Definition~\ref{def:diagonalizable}). Finally, also Propositions \ref{prop:diagonalizableboson} and \ref{prop:diagonalizablefermion} for diagonalizability carry over from $ \widehat{\sH} $ to $ \FB $, with the restriction to finitely many particle--hole transformed modes in the fermionic case.\\

\end{appendices}
\bigskip

\noindent\textbf{Acknowledgments.}
I am grateful to Micha\l\; Wrochna, Roderich Tumulka, Andreas Deuchert, Jean-Bernard Bru and Niels Benedikter for helpful discussions.\\
This work was financially supported by the DAAD (Deutscher Akademischer Austauschdienst) and also by the Basque Government through the BERC 2018-2021 program and by the Ministry of Science, Innovation and Universities: BCAM Severo Ochoa accreditation SEV-2017-0718, as well as by the European Union (ERC FermiMath, grant agreement nr. 101040991). Views and opinions expressed are however those of the author(s) only and do not necessarily reflect those of the European Union or the European Research Council Executive Agency. Neither the European Union nor the granting authority can be held responsible for them. Moreover, it was partially supported by Gruppo Nazionale per la Fisica Matematica in Italy.\\

\noindent\textbf{Data Availability.}
No datasets were generated during this work.\\

\section*{Declarations}

\noindent\textbf{Conflict of Interest.}
The author has no conflicts to disclose.\\

\end{document}